\documentclass[11pt]{article}

\usepackage[margin=1in]{geometry}

\usepackage{xcolor}
\definecolor{navy}{rgb}{0, 0, 0.75}

\usepackage{graphicx}
\usepackage{tikz}

\usepackage{caption}
\usepackage{subcaption}

\usepackage{amsmath, amsthm, amssymb}
\usepackage{bm}

\usepackage[shortlabels]{enumitem}

\usepackage{algorithm}
\usepackage[noend]{algpseudocode}

\usepackage[colorlinks, allcolors=navy]{hyperref}
\usepackage[capitalize,nosort]{cleveref}

\numberwithin{equation}{section}
\theoremstyle{plain}
\newtheorem{theorem}[equation]{Theorem}
\newtheorem{corollary}[equation]{Corollary}
\newtheorem{lemma}[equation]{Lemma}
\theoremstyle{definition}
\newtheorem{definition}[equation]{Definition}

\usepackage{thmtools}
\usepackage{thm-restate}

\DeclareMathOperator*{\E}{\mathbb{E}}

\newcommand{\iid}{\overset{\textit{iid}}{\sim}}

\newcommand{\eps}{\varepsilon}
\newcommand{\abs}[1]{\lvert#1\rvert}
\newcommand{\Abs}[1]{\left\lvert#1\right\rvert}
\newcommand{\norm}[1]{\lVert#1\rVert}

\newcommand{\R}{\mathbb{R}}
\newcommand{\N}{\mathbb{N}}
\newcommand{\Z}{\mathbb{Z}}
\newcommand{\cA}{\mathcal{A}}
\newcommand{\cB}{\mathcal{B}}
\newcommand{\cC}{\mathcal{C}}
\newcommand{\cD}{\mathcal{D}}
\newcommand{\cF}{\mathcal{F}}
\newcommand{\cG}{\mathcal{G}}

\newcommand{\cP}{\mathcal{P}}
\newcommand{\cQ}{\mathcal{Q}}
\newcommand{\cR}{\mathcal{R}}
\newcommand{\cS}{\mathcal{S}}
\newcommand{\cT}{\mathcal{T}}
\newcommand{\cU}{\mathcal{U}}
\newcommand{\cV}{\mathcal{V}}

\newcommand{\cX}{\mathcal{X}}
\newcommand{\cY}{\mathcal{Y}}

\newcommand{\hcD}{\widehat{\mathcal{D}}}
\newcommand{\tcD}{\widetilde{\mathcal{D}}}
\newcommand{\TV}{d_{\mathrm{TV}}}

\newcommand\blfootnote[1]{
    \begingroup \renewcommand\thefootnote{}\footnote{#1}
    \addtocounter{footnote}{-1}
    \endgroup
}

\title{Efficient and Private Property Testing via Indistinguishability}
\author{
    Cynthia Dwork \\
    Harvard University \\
    \url{dwork@seas.harvard.edu}
\and
    Pranay Tankala \\
    Harvard University \\
    \url{pranay_tankala@g.harvard.edu}
}
\date{April 6, 2026}

\begin{document}

\maketitle

\thispagestyle{empty}

\begin{abstract}
    Given a small random sample of $n$-bit strings labeled by an unknown Boolean function, which properties of this function can be tested computationally efficiently? We show an equivalence between properties that are efficiently testable from few samples and properties with \emph{structured symmetry}, which depend only on the function's average values on an efficiently computable partition of the domain. Without the efficiency constraint, a similar characterization in terms of unstructured symmetry was obtained by Blais and Yoshida (2019). We also give a function testing analogue of the classic characterization of testable graph properties in terms of regular partitions, as well as a sublinear time and \emph{differentially private} algorithm to compute concise summaries of such partitions of graphs. Finally, we tighten a recent characterization of the computational indistinguishability of product distributions, which encompasses the related task of efficiently testing which of \emph{two} candidate functions labeled the observed samples.

    Essential to our proofs is the following observation of independent interest: Every randomized Boolean function, no matter how complex, admits a \emph{supersimulator}: a randomized polynomial-size circuit whose output on random inputs cannot be efficiently distinguished from reality with constant advantage, {\em even by polynomially larger distinguishers}. This surprising fact is implicit in a theorem of Dwork et al. (2021) in the context of algorithmic fairness, but its complexity-theoretic implications were not previously explored. We give a new proof of this lemma using an iteration technique from the graph regularity literature, and we observe that a subtle quantifier switch allows it to powerfully circumvent known barriers to improving the landmark \emph{complexity-theoretic regularity lemma} of Trevisan, Tulsiani, and Vadhan (2009).
\end{abstract}

\blfootnote{This work was supported in part by Simons Foundation Grant 733782 and Cooperative Agreement CB20ADR0160001 with the United States Census Bureau.}

\newpage
\thispagestyle{empty}
\tableofcontents
\newpage
\pagestyle{plain}
\setcounter{page}{1}

\section{Introduction}
\label[section]{sec:intro}

The goal in distribution testing \cite{goldreich1996property, batu2000testing} is to determine whether an unknown probability distribution $\cD$ has a particular property of interest, such as uniformity or equality to a reference distribution. Rather than being given a complete description of $\cD$, the algorithm has access to an oracle that generates independent samples from $\cD$ on request. Ideally, we would like to make as few requests as possible, run a fairly simple computation on the resulting samples, and determine whether $\cD$ has the property or is far from having the property.

In this work, we focus on distributions over pairs $(x, y) \in \{0, 1\}^n \times \{0, 1\}$, where $x$ is a uniform $n$-bit string, and $y = f(x)$ for some unknown Boolean function $f : \{0, 1\}^n \to \{0, 1\}$. Our main result is a characterization of which properties of $f$ can be efficiently tested from few samples in this framework. By ``efficient,'' we mean that the testing algorithm can be represented as a small Boolean circuit that takes as input a set of labeled samples and outputs \textsc{Accept} or \textsc{Reject}. We also design a sublinear time \emph{differentially private}  \cite{dwork2006calibrating} algorithm to compute concise summaries of \emph{regular partitions} of graphs, which are known to characterize testability for graph properties \cite{alon2006regularity}, and give a function testing version of the characterization. Finally, we tighten a recent characterization of the computational indistinguishability of product distributions, which corresponds to the conceptually easier task of \emph{simple hypothesis testing}.

\subsection{Background}

Without the requirement of computational efficiency, an elegant result of Blais and Yoshida fully characterizes the properties of Boolean functions that are testable from a constant number of samples (i.e. not scaling with $n$) \cite{blais2019testable}. They showed that a property is constant-sample testable if and only if it has constant-part \emph{symmetry}, where a property of $f$ is said to have \emph{$k$-part symmetry} if permutations of the domain within parts of some fixed partition $S_1, \ldots, S_k \subseteq \{0, 1\}^n$ do not affect whether or not $f$ has the property. Equivalently, the property depends only on the average values of $f$ within each part.

More formally, let $\cP$ be a property, which we represent as a set of Boolean functions. Let $\cP_\eps$ denote the set of functions that are \emph{$\eps$-close} to some function in $\cP$, where distance is measured by the fraction of inputs on which two functions disagree. We say that a tester has \emph{proximity parameter $\eps$} if it accepts all $f \in \cP$ with probability at least $2/3$ and rejects all $f \notin \cP_\eps$ with probability at least $2/3$ (see \cref{sec:preliminaries} for more detail on the setup).  In this language, Blais and Yoshida showed:

\begin{theorem}[\cite{blais2019testable}]
\label[theorem]{thm:blais}
    If a property $\cP$ of Boolean functions is testable with proximity $\eps$ using $m$ samples, then $\cP \subseteq \cQ \subseteq \cP_{\eps}$ for some $2^{2^{O(m)}}$-part symmetric property $\cQ$.\footnote{Note that $m$ may scale with $n$ and $\eps$ arbitrarily, but the conclusion will be vacuous if $2^{2^{O(m)}}$ exceeds $2^n$.} Conversely, any $k$-part symmetric property $\cQ$ is testable with proximity $\eps$ using $(k/\eps)^{O(1)}$ samples.
\end{theorem}

Theorem~\ref{thm:blais} gives a transformation from $\cP$ to $\cQ$ that loses computational efficiency (ours will not). Indeed, $\cQ$ may be computationally intractable to test even when $\cP$ has a computationally efficient tester. 
To illustrate this, consider any $k$-part symmetric property $\cQ$ with label-invariant sets $S_1, \ldots, S_k$. Clearly, $\cQ$ can be tested from a handful of $(x_i, y_i)$ samples---just classify each $x_i$ according to the part $S_j$ that contains it, and use the corresponding labels $y_i$ to empirically estimate the sum of $f$ over $S_j$. Given these $k$ sums, testing for $\cQ$ can always be done efficiently, simply because any function on few inputs can be computed with a small circuit by brute force. However, the overall procedure may still be computationally inefficient---indeed, because the sets $S_j$ are unstructured, performing even a single classification could require up to $\exp(n)$ time!

\subsection{Results Overview}

The previous example raises a natural question: does there exist an analogue of \cref{thm:blais} that characterizes the properties of Boolean functions that are \emph{efficiently} testable from a constant number of samples? Our first result shows that this is indeed the case. We prove that if the tester for $\cP$ is a small circuit, then $\cP$ is close to a property $\cQ$ with \emph{structured} symmetry. By this, we mean that $\cQ$ is not only $k$-part symmetric with respect to some partition $S_1, \ldots, S_k$, but also has a \emph{classifier} circuit of size at most $s$ that computes the index $j \in [k]$ of the part $S_j$ containing a given input $x \in \{0, 1\}^n$. In this case, we say that $\cQ$ has \emph{computational partition complexity $s$}, or simply \emph{partition complexity $s$}.

We first strengthen the forward direction of \cref{thm:blais}. While \cref{thm:blais} does not bound the partition complexity of $\cQ$, which may be as large as $\exp(n)$, the following theorem bounds the partition complexity by $2^{O(m)}s$, which is a potentially much smaller quantity.

\begin{restatable}{theorem}{mainhard}
\label[theorem]{thm:main-hard}
    If $\cP$ is testable with proximity $\eps$ using $m$ samples and a circuit of size $s$, then $\cP \subseteq \cQ \subseteq \cP_\eps$ for some $2^{2^{O(m)}}$-part symmetric property $\cQ$ with partition complexity $2^{O(m)}s$.
\end{restatable}

Next, we strengthen the reverse direction of \cref{thm:blais}.

\begin{restatable}{theorem}{maineasy}
\label[theorem]{thm:main-easy}
    Any $k$-part symmetric property $\cQ$ with partition complexity at most $s$ is testable with proximity $\eps$ using $(k/\eps)^{O(1)}$ samples and a circuit of size $(k/\eps)^{O(1)}s + (k/\eps)^{O(k)}$.
\end{restatable}

When combined, \cref{thm:main-hard,thm:main-easy} give a characterization of properties \emph{efficiently} testable from a constant number of samples. To illustrate this, suppose that $m$, $k$, and $\eps$ are constants, so that the circuit size $s$ is the only parameter varying with $n$. In this case, \cref{thm:blais} states that testability is equivalent to symmetry, while our \cref{thm:main-hard,thm:main-easy} state that testability with a circuit of size $O(s)$ is equivalent to symmetry with partition complexity $O(s)$. We prove these theorems in \cref{sec:symmetry-thm}. Of the two, \Cref{thm:main-hard} is the harder to prove. Its proof uses a tool that we call \emph{supersimulators}, which we discuss further in the technical overview in \cref{sec:tech}. At a high level, the tool is a strengthening of the \emph{complexity-theoretic regularity lemma} of \cite{trevisan2009regularity} that arose implicitly in the algorithmic fairness literature. Thus, \cref{thm:main-hard} and its proof can be viewed as situated at the intersection of property testing, pseudorandomness (specifically, the structure-vs-randomness paradigm exemplified by the regularity lemma), and algorithmic fairness.

\Cref{thm:blais,thm:main-hard,thm:main-easy} characterize not only \emph{constant}-sample testability but also, for example, testability with $m(n)$ samples where $m(n) = O(\log^{(k)}(n))$ for all $k \in \N$ and $2^n$ is the domain size. Here, $\log^{(k)}$ denotes the $k$-fold iterated logarithm, so this class includes, for example, the function $m(n) = \log^*(n)$. In fact, even for properties testable from as many as $O(\log n)$ samples, the more interesting direction of our equivalence---efficient testing implies structured symmetry---continues to hold. Our next result pushes further in this direction, shedding light on the structure of properties that may require even more than $O(\log n)$ samples to test. Specifically, we present a modification to \cref{thm:blais} that circumvents the $2^{2^{O(m)}}$ dependence on the tester's sample size $m$, which is clearly a prerequisite for making any meaningful statement about sample sizes larger than $O(\log n)$. Specifically, we will show in \cref{sec:cc} that any $m$-sample testable property essentially boils down to what we call \emph{consistency counting} on a collection of $2^{O(m)}$ functions.

\begin{restatable}[Consistency Counting]{definition}{defcc}
\label[definition]{def:cc}
    We say that a collection of $m$ samples $(x_i, y_i)$ is \emph{consistent} with a function $f$ if the label $y_i = f(x_i)$ for all indices $i \in [m]$. Given a family of \emph{good} functions $\cF_+$ and a family of \emph{bad} functions $\cF_-$, the \emph{$(m, \cF_+, \cF_-)$-consistency counter} is the tester that receives $m$ inputs and outputs \textsc{Accept} iff they are consistent with strictly more good functions than bad.
\end{restatable}

Consistency counters are especially simple testers, which just check whether each of the prespecified functions in $\cF_+ \cup\cF_-$ could have plausibly labeled the $m$ observed samples. The following theorem states that every property that is testable from few samples has a consistency counter using a similar number of samples. From another point of view, it gives a \emph{two-sided error} analogue of the following folklore fact regarding one-sided sample-based property testing: If a property is testable with one-sided error, then it is testable with the ``canonical'' one-sided tester which exhaustively checks whether all observed samples are consistent with at least one function with the property.

\begin{restatable}{theorem}{thmcc}
\label[theorem]{thm:cc}
    If a property $\cP$ is testable with proximity $\eps$ using $m$ samples, then it is also testable with proximity $\eps$ using some $(O(m), \cF_+, \cF_-)$-consistency counter with $\abs{\cF_+ \cup \cF_-} \le 2^{O(m)}$.
\end{restatable}

There is a straightforward converse to \cref{thm:cc}: By definition, the $(m, \cF_+, \cF_-)$-consistency counter uses just $m$ samples, so if a property is testable with such a consistency counter, then it is testable with $m$ samples.

So far, we have framed our discussion around sample-testing Boolean functions, but in some sense, this framing is more restrictive than necessary. We remark that our main result, \cref{thm:main-hard}, can be straightforwardly extended to properties of distributions on an arbitrary domain $\cX$, provided that the distribution to be tested is known in advance to have high \emph{density} or, equivalently, high \emph{min-entropy}. We state the extension in terms of \emph{$\mu$-densely testable} properties, whose formal definition we defer to \cref{sec:dt}.

\begin{restatable}{theorem}{distrhard}
\label[theorem]{thm:distr-hard}
    Let $0 < \mu \le 1/2$. If a distribution property $\cP$ is $\mu$-densely testable with proximity $\eps$ using $m$ samples and a circuit of size $s$, then $\cP \subseteq \cQ \subseteq \cP_\eps$ for some $2^{(1/\mu)^{O(m)}}$-part symmetric distribution property $\cQ$ with partition complexity $(1/\mu)^{O(m)}s$.
\end{restatable}

\subsubsection{Graph Property Testing}
\label[section]{sec:intro-privacy}

Our next several results are motivated by the classic characterization of testable graph properties of Alon et al. \cite{alon2006regularity}. In that work, it was shown that a graph property is testable from few edge queries if and only if the property is roughly equivalent to having a \emph{regular partition} compatible with one of finitely many prespecified \emph{templates}. A ``regular partition'' refers to the kind provided by the \emph{Szemer\'{e}di regularity lemma}: a partition of the graph's vertices into a small number of parts between which the graph's edges are distributed pseudorandomly \cite{szemeredi1975regular}. Roughly speaking, a ``template'' is represented by a list of the pairwise densities between pairs of such parts.

Our first result in this setting is an extremely efficient (sublinear time) \emph{node-level differentially private} algorithm to compute templates corresponding to such regular partitions. Although our algorithm is extremely natural, we will need to leverage nontrivial results from \cite{alon2006regularity} for its analysis. Our second result in this setting will be an analogue of the characterization for efficiently testable Boolean function properties, as opposed to graph properties.

To state our results on private graph testing, we must first review some key notation regarding regular partitions of graphs. For convenience, we focus on simple graphs $G = (V, E)$. Given vertex sets $A, B \subseteq V$, we write $e_G(A, B) = \abs{E \cap (A \times B)}$ for the number of edges from $A$ to $B$. We write $d_G(A, B) = e_G(A,B)/\abs{A}\abs{B}$ for the edge density from $A$ to $B$. We omit the subscript $G$ from $e_G$ and $d_G$ when clear from context. A partition $(V_1, \ldots, V_k)$ of the vertices $V$ is an \emph{equipartition} if any two parts $V_i, V_j$ differ in size by at most $1$. We denote the total number of vertices by $n = \abs{V}$, and we write $G[Q]$ for the subgraph of $G$ induced by a vertex subset $Q \subseteq V$.

\begin{definition}[Graph Regularity \cite{szemeredi1975regular}]
\label[definition]{def:graph-regularity}
    A pair of vertex sets $(A, B)$ in a graph $G = (V, E)$ is \emph{$\gamma$-regular} if for all subsets $A' \subseteq A$ and $B' \subseteq B$ satisfying $\abs{A'} \ge \gamma \abs{A}$ and $\abs{B'} \ge \gamma \abs{B}$,
    \[
        \bigl\lvert d(A',B')-d(A,B) \bigr\rvert \le \gamma.
    \]
    An equipartition $(V_1, \ldots, V_k)$ is \emph{$\gamma$-regular} if at most $\gamma k^2$ of the pairs $(V_i, V_j)$ fail to be $\gamma$-regular.
\end{definition}

\begin{theorem}[\cite{szemeredi1975regular}]
\label[theorem]{thm:sz}
    Every graph has a $\gamma$-regular equipartition of size $k = O_\gamma(1)$.\footnote{The suppressed constant is a power tower of $2$s of height $\mathrm{poly}(1/\gamma)$, which is large but does not scale with $n$.}
\end{theorem}

The main result of \cite{alon2006regularity} is a characterization in terms of regular partitions of testable properties in the \emph{dense graph model}. In this model, properties are required to be invariant under vertex relabeling, and the tester accesses the graph via edge membership queries (i.e. ``is an edge present between these vertices $u$ and $v$?''). The distance between two graphs on the same vertex set is the fraction of edges on which they differ, or, more precisely, their edit distance under edge additions/deletions, divided by $n^2$. The characterization of testability is stated in terms of \emph{regularity templates} and \emph{regular-reducibility}:

\begin{definition}[Regularity Template \cite{alon2006regularity}]
\label[definition]{def:regularity-template}
    A \emph{regularity template} $R$ is defined by parameters $\gamma \in (0, 1]$ and $k \in \N$, a function $\eta : [k]^2 \to [0, 1]$, and a set $\bar{R} \subseteq [k]^2$ of size at most $\gamma k^2$. A graph $G$ \emph{satisfies} $R$ if it has an equipartition $(V_1, \ldots, V_k)$ such that $d(V_i, V_j) = \eta_{ij}$ for all pairs $(i, j) \in [k]^2$ and $(V_i, V_j)$ is $\gamma$-regular for all pairs $(i, j) \notin \bar{R}$.\footnote{As in \cite{alon2006regularity}, when we say $d(V_i, V_j) = \eta_{ij}$, we really mean that $e(V_i, V_j) = \lfloor \eta_{ij} \abs{V_i}\abs{V_j} \rfloor$.} We say that $G$ is \emph{$\eps$-close to satisfying} $R$ if $G$ is $\eps$-close to some graph $G'$ satisfying $R$. The \emph{complexity} of $R$ is $\max(k, 1/\gamma)$.
\end{definition}

\begin{definition}[Regular-Reducible \cite{alon2006regularity}]
\label[definition]{def:regular-reducible}
    A property $\cP$ of $n$-vertex graphs is \emph{regular-reducible} if for all $\eps > 0$, there exists a set $\cR$ of $O_\eps(1)$ templates, each of complexity $O_\eps(1)$, such that $\cP \subseteq \cQ_\eps$ and $\cQ \subseteq \cP_\eps$, where $\cQ$ denotes the property of satisfying some template in $\cR$.\footnote{The choice of which regularity templates belong to the set $\cR$ is allowed to depend on $n$, as long the number of templates and their complexity remain bounded by a constant $O_\eps(1)$ independent of $n$.}\footnote{The condition ``$\cP \subseteq \cQ_\eps$ and $\cQ \subseteq \cP_\eps$'' is slightly weaker than the condition $\cP \subseteq \cQ \subseteq \cP_\eps$ that appears elsewhere in this paper. The latter ensures that any tester for $\cQ$ is also a tester for $\cP$. The former only ensures that any \emph{tolerant} tester for $\cQ$ is also a tester for $\cP$. In the case of graph property testing, however, tolerant and non-tolerant testing are known to be roughly equivalent \cite{fischer2005estimation,gishboliner2023estimation}.}
\end{definition}

\begin{theorem}[\cite{alon2006regularity}]
\label[theorem]{thm:non-dp-graph-testing}
    A property $\cP$ of $n$-vertex graphs is regular-reducible if and only if for all $\eps > 0$, the property $\cP$ is testable with proximity $\eps$ using $O_\eps(1)$ edge queries.
\end{theorem}

Given the importance of regular partitions of graphs, a natural question is whether such summaries can be computed efficiently, ideally in time sublinear in the number of vertices $n$. When data privacy is not a concern, the answer turns out to be yes: one particularly elegant and simple algorithm is to take a uniformly random sample of ``anchor'' vertices, fully reveal their neighborhoods, and use the intersections of these neighborhoods to partition the graph \cite{tao2009random}. However, this approach may be unsatisfactory in real-world networks, where privacy is a concern and the presence or absence of edges between nodes may be sensitive information. Our first result in this section is a modified algorithm for efficiently computing regularity templates while satisfying \emph{differential privacy (DP)} \cite{dwork2006calibrating}. In fact, we go beyond edge privacy to the challenging setting of \emph{node-level} differential privacy, whose formal definition we defer to \cref{sec:preliminaries}.

\begin{restatable}{theorem}{dppartition}
\label[theorem]{thm:dp-partition}
    For all $\gamma, \eps > 0$, there exist parameters $k, q = O_{\gamma, \eps}(1)$ and $\lambda = \Omega_{\gamma,\eps}(1)$ such that for all input graphs $G$ on $n$ vertices, with probability at least $3/5$, \cref{alg:dp-partition} with these parameters outputs a regularity template $R$ with $\gamma$ error and $k$ parts such that $G$ is $\eps$-close to satisfying $R$. Moreover, \cref{alg:dp-partition} satisfies node-level $\alpha$-differential privacy for $\alpha = O_{\gamma,\eps}(1/n)$.
\end{restatable}

\begin{algorithm}
    \caption{Node-Level Differentially Private Regular Partition}
    \label{alg:dp-partition}
    \begin{algorithmic}[1]
        \State \textbf{input:} $G = (V, E)$ and $k, q \in \N$ and $\gamma, \lambda \in (0, 1)$
        \State $Q \leftarrow \text{uniformly random subset of $V$ of size $q$}$
        \State $\cR \leftarrow \text{set of all templates $(\gamma, k, \eta, \bar{R})$ such that $\eta_{ij}/\lambda \in \Z$ for all $(i, j) \in [k]^2$}$
        \ForAll{$R \in \cR$}
            \State $\ell(R) \leftarrow \text{fewest edge additions to (or deletions from) $G[Q]$ required to satisfy $R$}$
        \EndFor
        \State $R' \leftarrow \text{single sample from $\cR$ with probability proportional to $\exp(-\ell(R')/q)$}$
        \State \Return $R'$
    \end{algorithmic}
\end{algorithm}

Although we defer the full analysis of \cref{alg:dp-partition} to \cref{sec:dp-graphs}, the idea behind it is very simple. First, we take a random subsample of $q$ vertices from the graph, and we perform all $\binom{q}{2}$ edge queries between pairs of these vertices to obtain the induced subgraph $G[Q]$. Then, we exhaustively search over all possible regularity templates of a fixed granularity. We assign each candidate template $R$ a score depending on how close $G[Q]$ is to satisfying $R$, and then sample one such $R$ from a probability distribution weighted according to these scores. This is the classic \emph{exponential mechanism} \cite{mcsherry2007mechanism} from the DP literature.

Ultimately, since the number of templates is independent of $n$, our algorithm runs in constant $O_{\gamma, \eps}(1)$ time, not counting the time required to sample $q = O_{\gamma, \eps}(1)$ vertices from $V$. The fact that analyzing such a small number of vertices suffices to determine the partition structure of the overall graph is highly nontrivial, and it will require us to leverage results from \cite{alon2006regularity}. We emphasize that what makes \cref{alg:dp-partition} interesting is its efficiency. For example, if one were concerned \emph{only} with privacy and approximation quality, but not efficiency, one could simply post-process any state-of-the-art private cut sparsifier, which are known to be computable in $\mathrm{poly}(n)$ time with error subsumed by that of the regularity construction \cite{gupta2012iterative,arora2019sparsification,marek2020synthetic,liu2024optimal,zou2025almost,amand2025breaking}. In contrast, our algorithm takes sublinear time.

We also extend the characterization of testability of \cref{thm:non-dp-graph-testing} from the non-private to the private setting. The proof of the following theorem, which appears in \cref{sec:dp-graphs}, is based on a reduction from the private case to the non-private case, although it can also be proved directly using techniques similar to our analysis of \cref{alg:dp-partition}.

\begin{restatable}{theorem}{dptesting}
\label[theorem]{thm:dp-graph-testing}
    A property $\cP$ of $n$-vertex graphs is regular-reducible if and only if for all $\eps > 0$, the property $\cP$ is testable with proximity $\eps$ using $O_\eps(1)$ edge queries and node-level $O_\eps(1/n)$-DP.
\end{restatable}

To conclude this section, we briefly shift our attention from graph testing back to function testing, with the goal of proving a function testing analogue of the (original, non-private) characterization from \cite{alon2006regularity}. In this setting, we show that if a property of Boolean functions is testable using $m$ samples and a Boolean circuit of size $s$, then the property is essentially assessing compatibility with a collection $\cT$ of appropriately defined regularity templates, each of complexity $O(m^2s)$. In particular, this result is nontrivial for all arbitrary sample sizes $m$ and circuit sizes $s$, as long as they are subexponential in $n$. In the Boolean circuit context, we shall model a template by a \emph{simulator} for the function to be tested, as per the \emph{complexity-theoretic regularity lemma} of \cite{trevisan2009regularity}. For us, \emph{compatibility} will be the usual notion of computational indistinguishability.\footnote{We say two functions $g,h$ are $(s, \delta)$-indistinguishable if $\abs{\E[f(x)(g(x)-h(x))]} \le \delta$ for all functions $f$ of complexity at most $s$. See \cref{def:indistinguishability} for more details.}

\begin{restatable}{definition}{compatibility}
\label[definition]{def:compatibility}
    Given a family $\cT \subseteq \{\{0, 1\}^n \to [0, 1]\}$ (the \emph{templates}), the \emph{compatibility} property $\cT_{s,\delta}$ is the set of Boolean functions that are $(s, \delta)$-indistinguishable from some function in $\cT$. 
\end{restatable}

The following result states that small-circuit testability, even for a moderate or large number of samples, reduces to testing compatibility with a size-bounded collection of regularity templates.

\begin{restatable}{theorem}{templateshard}
\label[theorem]{thm:templates-hard}
    If $\cP$ is testable with proximity $\eps$ using $m$ samples and a circuit of size $s$, then there exists a set $\cT$ of templates, each of circuit complexity $O(m^2s)$, such that $\cP \subseteq \cT_{s,\frac{1}{13m}} \subseteq \cP_\eps$.
\end{restatable}

We also prove a partial converse to the preceding theorem. We say ``partial'' because it only preserves sample efficiency, not computational efficiency.

\begin{restatable}{theorem}{templateseasy}
\label[theorem]{thm:templates-easy}
    Let $s \le t$ with $t \log(t) \le O(\eps^2\delta^2 2^n)$. If $\cT$ is a set of templates, each of circuit complexity at most $t$, then  $\cT_{s,\delta}$ is testable with proximity $\eps$ from $O(t\log(t)/\eps^2\delta^2)$ samples.
\end{restatable}

\subsubsection{Simple Hypothesis Testing}

All of the function testing results discussed so far concern the broad task of distinguishing a property $\cP$ from the complement of $\cP_\eps$, both of which may be doubly exponentially large sets of Boolean functions. Building on the recent work of \cite{marcussen2024characterizing}, our next few results concern the narrower task of \emph{simple hypothesis testing}, which requires distinguishing between just \emph{two} candidates.

For context, consider the following slightly more general question: Given $k$ independent samples from a distribution $\cD_b \in \{\cD_0, \cD_1\}$ over $\{0, 1\}^n$, where $b \in \{0, 1\}$ is an unknown bit, what is the best distinguishing advantage between $\cD_0$ and $\cD_1$ that one could hope to efficiently achieve? With a computationally unbounded distinguisher, the answer is the \emph{total variation distance} between the $k$-fold products of $\cD_0$ and $\cD_1$, denoted $d_{\mathrm{TV}}(\cD_0^{\otimes k}, \cD_1^{\otimes k})$. It is easy to show that this quantity relates to $d_{\mathrm{TV}}(\cD_0, \cD_1)$ via the following standard inequalities:
\[
    d_{\mathrm{TV}}(\cD_0^{\otimes k}, \cD_1^{\otimes k}) \le 1 - (1 - d_{\mathrm{TV}}(\cD_0, \cD_1))^k \le k \cdot d_{\mathrm{TV}}(\cD_0, \cD_1).
\]
If we let $d_s(\cdot,\cdot)$ denote the best distinguishing advantage attainable by a circuit of size at most $s$, then it turns out that proving the analogous relationship between $d_s(\cD_0^{\otimes k}, \cD_1^{\otimes k})$ and $d_s(\cD_0, \cD_1)$ is significantly more involved \cite{halevi2008degradation, geier2022tight}.

To facilitate reasoning about the $k$-fold product, it would be ideal if each $\cD_b$ had a computationally indistinguishable proxy distribution $\widetilde{\cD}_b$ such that $d_s(\cD_0^{\otimes k}, \cD_1^{\otimes k}) \approx d_{\mathrm{TV}}(\widetilde{\cD}_0^{\otimes k}, \widetilde{\cD}_1^{\otimes k})$. A recent result \cite{marcussen2024characterizing} based on \emph{multicalibration} \cite{hkrr} achieves almost exactly this, but requires a second size bound $s' = O(sk/\eps^{6}) + (k/\eps)^{O(1)}$ to make the connection bidirectional. Our first result in this setting improves the leading exponent on $1/\eps$ from $6$ to $2$ using \emph{calibrated multiaccuracy}, a weaker but more efficiently achievable condition than multicalibration that has been a theme of several recent works, as we discuss further in \cref{sec:related}.

\begin{restatable}{theorem}{mpvcma}
\label[theorem]{thm:mpv-cma-succinct}
    Given $s, k \in \N$, $\eps > 0$ and $\cD_0, \cD_1$, there exist $\widetilde{\cD}_0, \widetilde{\cD}_1$ such that $d_s(\cD_b, \widetilde{\cD}_b) \le \eps$ and
    \[
        d_{s}(\cD_0^{\otimes k}, \cD_1^{\otimes k}) - k\eps \le d_{\mathrm{TV}}(\widetilde{\cD}_0^{\otimes k}, \widetilde{\cD}_1^{\otimes k}) \le d_{s'}(\cD_0^{\otimes k}, \cD_1^{\otimes k}) + k\eps,
    \]
    where $s' = O(sk/\eps^2) + (k/\eps)^{O(1)}$. One can enforce $\widetilde{\cD}_0 = \cD_0$ with $s' = O(sk/\eps^4) + (k/\eps)^{O(1)}$.
\end{restatable}

Much like the main result of \cite{marcussen2024characterizing}, \cref{thm:mpv-cma-succinct} is powerful because it allows one to translate statistical or information-theoretic arguments about $k$ independent draws from $\widetilde{\cD}_0$ or $\widetilde{\cD}_1$ to computational statements about $\cD_0$ and $\cD_1$. The advantage of our theorem compared to that of \cite{marcussen2024characterizing} lies in the smaller complexity gap ($1/\eps^{2}$ versus $1/\eps^{6}$, or $1/\eps^4$ versus $1/\eps^{12}$ for the single-proxy version), which measures the quantitative relationship between $s$ and $s'$.

Our second result in this setting eliminates the complexity gap entirely using supersimulators, the same technical tool underlying the proof of our first result, \cref{thm:main-hard}. For the following result, we now allow ourselves to alter the distinguisher size bound $s$ within a certain interval, with the specific choice depending on the pair $\cD_0, \cD_1$. The reason for this will become clear shortly when we discuss supersimulators in more detail.

\begin{restatable}{theorem}{mpvsuper}
\label[theorem]{thm:mpv-super-succinct}
    Given $s, k \in \N$, $\eps > 0$, $\cD_0, \cD_1$, there exist $s', \widetilde{\cD}_0, \widetilde{\cD}_1$ such that $d_{s'}(\cD_b, \widetilde{\cD}_b) \le \eps$ and
    \[d_{s'}(\cD_0^{\otimes k}, \cD_1^{\otimes k}) - k\eps \le d_{\mathrm{TV}}(\widetilde{\cD}_0^{\otimes k}, \widetilde{\cD}_1^{\otimes k}) \le d_{s'}(\cD_0^{\otimes k}, \cD_1^{\otimes k}) + k\eps,\]
    where $s' \in [s, k^{O(1/\eps^2)}s]$. One can enforce $\tcD_0 = \cD_0$ with $s' \in [s, k^{O(1/\eps^4)}s]$.
\end{restatable}

The advantage of \cref{thm:mpv-super-succinct} is that the same term $d_{s'}(\cD_0^{\otimes k}, \cD_1^{\otimes k})$ appears on both the left and right side of the chain of inequalities. This yields a tighter characterization in terms of $d_{\mathrm{TV}}(\widetilde{\cD}_0^{\otimes k}, \widetilde{\cD}_1^{\otimes k})$ than in \cref{thm:mpv-cma-succinct}, in which distinct terms $d_{s}(\cD_0^{\otimes k}, \cD_1^{\otimes k}) \le d_{s'}(\cD_0^{\otimes k}, \cD_1^{\otimes k})$, which potentially differ significantly, are used for the left and right side of the chain of inequalities. We prove both \cref{thm:mpv-cma-succinct,thm:mpv-super-succinct} in \cref{sec:mpv}.

\subsection{Technical Overview}
\label[section]{sec:tech}

The results of this paper largely rely on a technical tool of independent interest, which we call \emph{supersimulators}. Although the existence of supersimulators is implicit in prior work \cite{dwork2021outcome}, the present work is the first to explore their complexity-theoretic implications. In \cref{sec:tech-super}, we motivate supersimulators, discuss their connection to graph regularity, and explain how, surprisingly, they circumvent known lower bounds to improving the complexity-theoretic regularity lemma of \cite{trevisan2009regularity}. In \cref{sec:tech-sketch}, we explain how they are used in the proof of our first main result, \cref{thm:main-hard}.

\subsubsection{Supersimulators and the Regularity Lemma}
\label[section]{sec:tech-super}

As previously discussed, the Szemer\'{e}di regularity lemma \cite{szemeredi1975regular} is a cornerstone result in graph theory that splits any large, dense graph---no matter how complex---into a small number of parts between which the graph's edges are distributed \emph{pseudorandomly}.  Although Szemer\'{e}di's regularity lemma spawned many variants, these versions were still fairly specialized (e.g. to cut sizes in graphs, or Fourier uniformity in vector spaces over finite fields) until the arrival of the \emph{complexity-theoretic regularity lemma} \cite{trevisan2009regularity}, which generalized and abstracted the concept of regularity by considering \emph{indistinguishability} with respect to an arbitrary collection of \emph{distinguisher} functions defined on an arbitrary domain.

For example, consider a setting especially relevant to applications in complexity theory and cryptography: distinguishers computable by size-bounded Boolean circuits that receive a uniformly random input from $\{0, 1\}^n$. In this setting, the regularity lemma states that every randomized Boolean function of arbitrary complexity can be \emph{simulated} by a randomized circuit of size $O(s)$ that fools all distinguishers of size at most $s$, for any $s$.

\begin{lemma}[Special Case of Theorem 1.1 of \cite{trevisan2009regularity}]
\label[lemma]{thm:ttv-for-circuits}
    For all target functions $g : \{0, 1\}^n \to [0, 1]$, sizes $s \in \N$, and error tolerances $\eps > 0$, there exists a simulator $h : \{0, 1\}^n \to [0, 1]$ of circuit size\footnote{We say a circuit $c$ with $n$ input bits, $m$ output bits, and $s$ logic gates computes a real-valued function $h : \{0, 1\}^n \to [0, 1]$ in binary if $h(x) = \sum_{i=1}^m  c_i(x) / 2^{i-1}$ for all inputs $x \in \{0, 1\}^n$, where $c_i$ denotes the $i$\textsuperscript{th} output bit of $c$.} at most $O(s/\eps^2)$ such that for all distinguishers $A : \{0, 1\}^{n+1} \to \{0, 1\}$ of size at most $s$,
    \[
        \Bigl\lvert{\scalebox{1.2}{$\mathbb{E}$}_{\substack{x \sim \{0, 1\}^n \\ y|x \sim \cB(g(x))}} \big[A(x, y)\bigr] - \scalebox{1.2}{$\mathbb{E}$}_{\substack{x \sim \{0, 1\}^n \\ y|x \sim \cB(h(x))}} \bigl[A(x, y)\bigr]\Bigr\rvert}\le \eps.\footnote{$\cB(p) \in \Delta(\{0, 1\})$ denotes the Bernoulli distribution with parameter $p$.}
    \]
\end{lemma}

\Cref{thm:ttv-for-circuits} ensures the existence of a simulator that is at most
a constant factor larger than the distinguishers, but makes no guarantees about the existence of \textit{smaller} simulators.  At least three different arguments in the literature, two in the original paper \cite{trevisan2009regularity} and one in a subsequent work \cite{chen2018complexity}, assert that this complexity gap is inevitable, and the simulator must always be allowed to be larger than the distinguishers that it is asked to fool. For more on these barriers, see the discussion of related work in \cref{sec:related}.

Surprisingly, all three lower bound arguments can be circumvented by relaxing the requirement that the result hold for all size bounds $s$.  By instead allowing the choice of $s$ to depend on the target function, while remaining below an upper bound independent of the target function, it is possible to design simulators that fool families of distinguishers far more powerful than themselves. In this work, we will refer to such an object as a \emph{supersimulator}. As our main results demonstrate, this relaxed quantification in which $s$ depends on $g$ is still useful for applications.

The existence of supersimulators is implicit in the \emph{code-access outcome indistinguishability} construction of \cite{dwork2021outcome}, in which the distinguishers to be fooled have access to the code of the simulator (and are therefore perforce larger). Thus, the following statement can be viewed as a rephrasing of Theorem 5.10 of \cite{dwork2021outcome}, unencumbered by questions of the distinguishers' access to the code of the simulator. A related result appeared even earlier, albeit in a slightly different computational model, in the study of key-agreement in cryptography \cite{haitner2018dichotomy}.

\begin{lemma}[Supersimulators, Special Case \cite{dwork2021outcome}]
\label[lemma]{thm:super-succinct}
    For all $g : \{0, 1\}^n \to \{0, 1\}$, $k \in \N$, and $\eps > 0$, there exists a size bound $s \in [n, n^{k^{O(1/\eps^2)}}]$ and a simulator $h : \{0, 1\}^n \to [0, 1]$ of size at most $s$ such that for all distinguishers $A : \{0, 1\}^{n+1} \to \{0, 1\}$ of size at most $s^{k}$,
    \[
        \Bigl\lvert{\scalebox{1.2}{$\mathbb{E}$}_{\substack{x \sim \{0, 1\}^n \\ y|x \sim \cB(g(x))}} \big[A(x, y)\bigr] - \scalebox{1.2}{$\mathbb{E}$}_{\substack{x \sim \{0, 1\}^n \\ y|x \sim \cB(h(x))}} \bigl[A(x, y)\bigr]\Bigr\rvert}\le \eps.
    \]
\end{lemma}

To understand the statement of \cref{thm:super-succinct}, suppose that $k = 100$, $\eps = 1/10$, and the average-case complexity of $g$ is superpolynomial in $n$. In this case, one cannot hope to approximate $g$ accurately with a function $h$ of complexity $s = \mathrm{poly}(n)$, but \cref{thm:super-succinct} nevertheless guarantees that there exists a simulator of this size fooling all distinguishers of size up to $s^{100}$, well beyond the capabilities of the simulator provided by \cref{thm:ttv-for-circuits}. The particular choice of $s$ may depend on $g$, but it never leaves the interval $[n, n^{10^{300}}]$, say, which is fixed and independent of $g$. For clarity, we state this concrete version separately, as a corollary:
\begin{corollary}
    For all $g : \{0, 1\}^n \to \{0, 1\}$, $k \in \N$, and $\eps > 0$, there exists $s = \mathrm{poly}(n)$ and $h : \{0, 1\}^n \to [0, 1]$ of size at most $s$ such that for all $A : \{0, 1\}^{n + 1} \to \{0, 1\}$ of size at most $s^{100}$,
    \[
        \Bigl\lvert{\scalebox{1.2}{$\mathbb{E}$}_{\substack{x \sim \{0, 1\}^n \\ y|x \sim \cB(g(x))}} \big[A(x, y)\bigr] - \scalebox{1.2}{$\mathbb{E}$}_{\substack{x \sim \{0, 1\}^n \\ y|x \sim \cB(h(x))}} \bigl[A(x, y)\bigr]\Bigr\rvert}\le \frac{1}{10}.
    \]
\end{corollary}

While we have stated both \cref{thm:ttv-for-circuits,thm:super-succinct} in the setting of size-bounded Boolean circuits on the domain $\cX = \{0, 1\}^n$, both results generalize to arbitrary distinguishers families on an arbitrary domain. We call this the \emph{abstract setting} and review it further in \cref{sec:preliminaries}.

In \cref{sec:supersimulators}, we present two proofs of \cref{thm:super-succinct}. The first more closely resembles the proof in \cite{dwork2021outcome}, as well as the proof of the original complexity-theoretic regularity lemma in \cite{trevisan2009regularity}, but we use a slightly modified projection operation which will be important for our downstream applications.

Our second proof is based on an iteration technique from the graph regularity literature, in which strong notions of regularity are achieved by iterating cheaper constructions with a shrinking sequence of error parameters. In graph theory, this technique has been used to achieve state-of-the-art quantitative bounds for applications to graph removal \cite{fox2011new, conlon2012bounds} and to establish relationships between existing notions of regularity \cite{rodl2010regularity}. In our setting, we use it to increase the power of the simulators, allowing them to fool distinguishers whose complexity surpasses their own by any specified ``growth function.'' This proof technique also has the advantage of a slightly stronger theorem statement, which incorporates a shrinking error tolerance that decays as a function of the simulator size. We will not, however, require this additional power for our applications to property testing.

\subsubsection{Finding Structured Symmetry}
\label[section]{sec:tech-sketch}

In this section, we sketch the proof of our first main result, \cref{thm:main-hard}. This result states that any property $\cP$ with an efficient tester $T$ is roughly equivalent to a property $\cQ$ with structured symmetry. Here, ``roughly equivalent'' means that $\cP \subseteq \cQ \subseteq \cP_\eps$.

Our approach will be to construct a simulator $\widetilde{T}$ for $T$ that has a particularly simple form, which we will then use to  define the property $\cQ$. Specifically, we shall say that a function $f$ has the property $\cQ$ if the tester $\widetilde{T}$ is more likely to output \textsc{Accept} than \textsc{Reject} when it receives $m$ independent inputs $(x_i, y_i)$ labeled by $f$. Here, $x_i \sim \{0, 1\}^n$ is a uniform $n$-bit string and $y_i = f(x_i)$. We will show that $\cP \subseteq \cQ \subseteq \cP_\eps$ by using the fact that $\widetilde{T}$ is a simulator for a tester for $\cP$. Similarly, we will argue that $\cQ$ has structured symmetry by reasoning about the simple form of $\widetilde{T}$.

Next, we describe how the simulator $\widetilde{T}$ is actually constructed. Our goal is for $\widetilde{T}$ to be a valid simulator for $T$, which means that $\E[\widetilde{T}(x, y)] \approx \E[T(x, y)]$ for any possible labeling function $f \in \cP$ (or $f \notin \cP_\eps$) that may have been used to label the input samples $(x_i, y_i)$.\footnote{Here, $T(x, y)$ is shorthand for $T((x_1, y_1), \ldots, (x_m, y_m))$, and $\widetilde{T}(x, y)$ is defined similarly.} To achieve this condition, it will be convenient to first replace the real labels $y_i$ with much simpler \emph{modeled} labels $\tilde{y_i}$, and instead argue that $\E[\widetilde{T}(x, \tilde{y})] \approx \E[T(x, \tilde{y})]$. While the real labels were defined by the equation $y_i = f(x_i)$, these modeled labels will be defined by the conditional probability distribution $\tilde{y}_i|x_i \sim \cB(\tilde{f}(x_i))$, where $\tilde{f}$ is itself a low-complexity simulator of $f$. For this step (simulating the sample-generating oracle), we will use the complexity-theoretic regularity lemma. For the other step (simulating the tester), we will need to use the stronger supersimulators lemma, for reasons that will soon become apparent. Intuitively, the reason is that the function $\tilde{f}$ and hence the modeled labels $\tilde{y}$ will be of slightly higher complexity than that of both $T$ and $\widetilde{T}$, so in order for $\widetilde{T}$ to fool them, it must be capable of fooling things slightly more powerful than itself.  Interestingly, between the two steps, the roles of the simulators and distinguishers switch.

\paragraph{Step 0: Notation} Before proceeding in more detail, it will be helpful to review the relevant lemmas using slightly simpler notation than before. Indeed, the complexity-theoretic regularity lemma can be viewed as stating that every function $g : \cX \to \{0, 1\}$ has a \emph{low-complexity simulator} $h : \cX \to [0, 1]$. By ``simulator,'' we mean that the error of $h$, namely $g-h$, is not too correlated with any function in a prespecified family $\cF$ of \emph{distinguisher} functions:
\[
    \max_{f \in \cF}\,\Bigl\lvert\scalebox{1.2}{$\mathbb{E}$}_{x \sim \cD}\bigl[f(x)\bigl(g(x)-h(x)\bigr)\bigr]\Bigr\rvert \le \delta.
\]
By ``low-complexity,'' we mean that $h$ is a weighted sum of a handful of functions from $\cF$:
\[
    h(x) = h_k(x) = \Bigl[\delta\cdot \bigl(f_1(x) + \cdots + f_k(x)\bigr)\Bigr]_0^1
\]
for some $f_1, \ldots, f_k \in \pm \cF$, where $[\cdot]_0^1$ denotes projection onto the interval $[0, 1]$. The supersimulators lemma replaces the fixed family $\cF$ with a variable family $\cF(h)$ which grows steadily more complex as $h$ grows in circuit size, meaning that for each index $j \in [k]$, we have $f_j \in \pm \cF(h_{j-1})$, where $h_{j-1}$ is defined analogously to $h_k$ above.

In what follows, we model the tester as a function $T : (\cX \times \{0, 1\})^m \to \{0, 1\}$ which takes as input $m$ labeled samples $(x_i, y_i) \in \cX \times \{0, 1\}$ and outputs either $0$ or $1$, corresponding to \textsc{Reject} or \textsc{Accept}. For simplicity, suppose for now that $T$ is deterministic. (We will handle randomized testers in the full version of the proof.) Also note that $\cX = \{0, 1\}^n$.

\paragraph{Step 1: Simulating the Oracle} First, we argue that if the property $\cP$ of interest has a tester $T$ computable by a small circuit of size $s$, then we need not consider arbitrarily complex labeling functions $f : \{0, 1\}^n \to \{0, 1\}$. Using the complexity-theoretic regularity lemma and a simple hybrid argument, we show that it suffices to study functions $\tilde{f} : \{0, 1\}^n \to [0, 1]$ that are computable by circuits of size $O(s)$, which is to say at most a constant factor larger than $T$. Now, the oracle simulator returns \emph{modeled} labels $\tilde{y}_i$ sampled from the Bernoulli distribution $\cB(\tilde{f}(x_i))$, rather than real labels $y_i = f(x_i)$, but they look no different to the tester. Slightly more formally, we suppose that $\tilde{f}$ is a $(\cF, \delta)$-regular simulator for $f$ with respect to the class $\cF$ of \emph{one-way restrictions} of $T$, which are functions obtained by hardwiring all but one input to $T$ (note that if $T$ has a circuit of size at most $s$, then so do functions in $\cF$). Under this assumption on $\tilde{f}$, we show that
\[
    \Bigl\lvert \E\bigl[T(x, y)\bigr] - \E\bigl[T(x, \tilde{y})\bigr] \Bigr\rvert \le 2m\delta.
\]
By the complexity-theoretic regularity lemma, there always exists such a simulator $\tilde{f}$ of circuit size $O(s)$, which, notably, is at most slightly larger than the circuit complexity of $T$.

\paragraph{Step 2: Simulating the Tester}

Next, consider two $m$-sample testers $T$ and $\widetilde{T}$. Suppose that they both receive labels generated from a circuit $\tilde{f}$ that is slightly larger than both testers, as suggested by the construction at the end of the previous step. Suppose also that $\tilde{f}$ happens to be $\{0, 1\}$-valued. (We will address the challenge of fractional values in the full version of the proof.) We show that $T$ and $\widetilde{T}$ have similar probabilities of outputting \textsc{Accept} if $T$ and $\widetilde{T}$ are indistinguishable by a certain function derived from $\tilde{f}$, called $\tilde{f}'$. This function receives $m$ labeled samples and checks whether they are all consistent with $\tilde{f}$:
\begin{equation}
\label{eq:tech-overview-f-prime}
    \tilde{f}'(x, y) = \bm{1}\bigl[\forall i \in [m],\, y_i = \tilde{f}(x_i) \bigr].
\end{equation}
Slightly more formally, we show that if $\widetilde{T}$ is a simulator for $T$ with error $\gamma$ with respect to the class of functions $\tilde{f}'$, where $\tilde{f}$ is any circuit of size at most $O(\max\{\mathrm{size}(T), \mathrm{size}(\widetilde{T})\})$, then
\[
    \Bigl\lvert \E\bigl[T(x, \tilde{y})\bigr] - \E\bigl[\widetilde{T}(x, \tilde{y})\bigr] \Bigr\rvert \le 2^m \gamma.
\]
(Recall from the previous step that $\tilde{y}_i = \tilde{f}(x_i)$ are modeled labels generated by $\tilde{f}$.)

Ideally, we would like to apply the complexity-theoretic regularity lemma once more to construct a small simulator $\widetilde{T}$ that fools $\tilde{f}'$ for any circuit $\tilde{f}$ that is slightly larger than $T$ and $\widetilde{T}$. This includes, in particular, the function $\tilde{f}$ constructed in the previous step, applied to both $T$ and $\widetilde{T}$. Indeed, if we could do this, then we would be done. We would simply define $\cQ$ to be the property of a function $f$ such that $f \in \cQ$ iff the tester $\widetilde{T}$ is more likely to accept than reject on labels generated by $f$. We would be able to argue that $\cP \subseteq \cQ \subseteq \cP_\eps$ using the relationship between the output of $T$ and $\widetilde{T}$ on labels $y$ generated by $f$:
\[
    \E\bigl[T(x, y)\bigr] \approx_{2m\delta} \E\bigl[T(x, \tilde{y})\bigr] \approx_{2^m \gamma} \E\bigl[\widetilde{T}(x, \tilde{y})\bigr] \approx_{2m\delta} \E\bigl[\widetilde{T}(x, y)\bigr].
\]
Moreover, setting $\delta = O(1/m)$ and $\gamma = O(1/2^m)$, the lemma would decompose $\widetilde{T}$ into a simple weighted sum of $\mathrm{poly}(1/\gamma) = 2^{O(m)}$ functions $\tilde{f}_1', \ldots, \tilde{f}_k'$. Taking all intersections of level sets of the low-complexity functions $\tilde{f}_1, \ldots, \tilde{f}_k$ would give us the efficiently computable, symmetric partition into $2^{2^{O(m)}}$ parts that we desire.

\paragraph{Step 3: Supersimulators}
There is a flaw in the preceding argument, as stated: We require the simulator $\widetilde{T}$ to fool distinguishers $\tilde{f}'$ that are slightly more complex than itself, which the complexity-theoretic regularity lemma cannot provide! In the full proof, we overcome this obstacle using supersimulators.

Unpacking definitions, we see that the supersimulator construction we employ is iterative, generating a sequence of approximations $\widetilde{T}_1, \widetilde{T}_2, \ldots$ to the original tester $T$. At the $j$\textsuperscript{th} step, we derive $\widetilde{T}_{j+1}$ from $\widetilde{T}_j$ by first choosing a function $\tilde{f}_j$ slightly larger than $\widetilde{T}_j$ and then adding an appropriate multiple of  $\tilde{f}_j'$ to $\widetilde{T}_j$. Specifically, the function $\tilde{f}_j : \cX \to [0, 1]$ will be a combination of several one-way restrictions of $\widetilde{T}_j$, which are obtained by hard-wiring all but one input to the current simulator $\widetilde{T}_j$, and the function $\tilde{f}_j' : (\cX \times \{0, 1\})^m \to \{0, 1\}$ is defined with respect to $\tilde{f}_j$ as per equation~\eqref{eq:tech-overview-f-prime}.

\subsection{Related Work}
\label[section]{sec:related}

Our paper extends a recent line of work investigating the interplay between classical results in pseudorandomness (specifically, the structure-vs-randomness paradigm) and recent results in the algorithmic fairness literature. In this work, we do so through the lens of property testing.

\paragraph{Structure vs Randomness}

There is a large body of work devoted to decomposing complex objects into their structured and pseudorandom components. As already discussed, an early result in this space is Szemer\'{e}di's regularity lemma \cite{szemeredi1975regular}. There is a vast literature on the regularity lemma and its variants, producing breakthrough results in pseudorandomness and additive combinatorics to this day \cite{kelley2023progressions,jaber2025corners}. For more, see the surveys \cite{tao2007structure,zhao2023gtac}.

The complexity-theoretic regularity lemma of \cite{trevisan2009regularity} gave a unified perspective on such disparate topics as the Frieze-Kannan weak regularity lemma for graphs \cite{frieze1996regularity, frieze1999approximation}, Impagliazzo's hardcore lemma \cite{impagliazzo1995hard}, and the dense model theorem \cite{green2008primes,tao2008primes,reingold2008dense}. It also led to a deeper understanding of computational entropy \cite{vadhan2012pseudoentropy,vadhan2013uniform,zheng2014thesis} and techniques for leakage simulation and key-agreement in cryptography \cite{jetchev2014fake, chen2018complexity, haitner2018dichotomy}. Similar ideas to those of \cite{trevisan2009regularity} also appeared in \cite{feldman2012fourier, de2012chow}.

\paragraph{The Regularity Barrier} In the Boolean circuit case, the regularity lemma (\cref{thm:ttv-for-circuits}) asserts that for any size bound $s \in \N$, there exists a simulator $h$ of size $s' = O(s)$ that fools distinguishers of size $s < s'$. The question of the existence of size-$s'$ simulators $h$ that fool distinguishers of size $s  \gg s'$ has already been considered by multiple works \cite{trevisan2009regularity, chen2018complexity}, with negative results.

First, Remark 1.6 of \cite{trevisan2009regularity} constructed two counterexamples ruling out the possibility of any version of \cref{thm:ttv-for-circuits} with $s \ge (ns')^{1 + \Omega(1)}$. One of their examples involves a target function $g : \{0, 1\} \to \{0, 1\}$ of complexity $\tilde{O}(n s')$ sampled from a family of $O(s' \log s')$-wise independent hash functions. Using a Chernoff-like concentration inequality for $k$-wise independence (e.g. Problem 3.8 of \cite{vadhan2012pseudorandomness}), they show that with high probability over the choice of $g$, every function $h$ of complexity at most $s'$ has correlation at most $1/10$ with $g$. Consequently, one cannot hope for $h$ to fool the distinguisher $f = g$ of size $s = \tilde{O}(ns')$.

Crucially, in the aforementioned counterexample, the target function $g$ depends on the simulator size $s'$ (for example, in one standard construction of the hash family, $g$ would be a polynomial whose degree grows with $s'$). The same is true of their other counterexample, which is based on a black-box application of randomness extractors for efficiently sampleable distributions with high min-entropy. Therefore, these counterexamples do not rule out supersimulators (\cref{thm:super-succinct}), which only guarantee that for all $g$, there exists at least one ``good'' simulator size $s'$ in a bounded interval.

The subsequent work of \cite{chen2018complexity} provided an even stronger lower bound on the simulator size, under additional assumptions on its structure. Specifically, they argue that any ``black-box'' $(\cF, \eps)$-regular simulator $h$ for $g$ under $\cD$ must make $\Omega(1/\eps^2)$ oracle calls to functions in $\cF$. In the Boolean circuit case, this result may lead one to believe that simulators of size $s'$ can only hope to fool \emph{substantially} smaller distinguishers, and not even those of slightly smaller sizes $s \in [\eps^2 s', s']$. From one point of view, \cref{thm:super-succinct} shows that in the context of Boolean circuits, removing this structural requirement on $h$ impacts the range of attainable simulator and distinguisher sizes.

\paragraph{Algorithmic Fairness}

In the context of algorithmic fairness for machine learning systems, modern concepts like \emph{outcome indistinguishability} \cite{dwork2021outcome} and \emph{multicalibration} \cite{hkrr,kearns2018gerrymandering}, studied further in work on \emph{omniprediction} \cite{omnipredictors, gopalan2023lossoi}, can be viewed as stronger, constructive versions of complexity-theoretic regularity that utilize practical learning-theoretic primitives.
\

The connection between these modern machine learning tools and older notions of regularity was made explicit by \cite{dwork2023pseudorandomness,casacuberta2024complexity}, and this point of view has proven to be fruitful. Indeed, these works and their \textit{sequelae} have led to new insights into graph regularity, hardcore set construction, dense models, omniprediction, computational hardness and entropy, and the computational indistinguishability of product distributions \cite{marcussen2024characterizing,casacuberta2025global,hu2025generalized}. The present work similarly builds on the ideas from this line of work, but now with a view toward property testing. In particular, \emph{supersimulation}, which is our main technical tool, is closely related to both outcome indistinguishability and multicalibration.

\paragraph{Calibrated Multiaccuracy} 
The idea behind our proof of \cref{thm:mpv-cma-succinct}, which concerns the computational indistinguishability of product distributions as in \cite{marcussen2024characterizing}, is to replace the use of multicalibration (a sometimes costly strengthening of the regularity lemma) with calibrated multiaccuracy (a lightweight strengthening of the regularity lemma). The use of this idea to achieve quantitative gains in downstream applications has been a theme of multiple recent and concurrent works.

One such work shows that the concept of omniprediction in machine learning can be achieved via calibrated multiaccuracy \cite{gopalan2023lossoi} (see also \cite{dwork2023pseudorandomness}), which had previously been shown to be achievable via multicalibration \cite{omnipredictors}. Another work shows that for hardcore set construction, optimal density can be achieved via calibrated multiaccuracy \cite{casacuberta2025global}, which had previously been shown to be achievable via multicalibration \cite{casacuberta2024complexity}, and even earlier with suboptimal density via the original regularity lemma \cite{trevisan2009regularity}. Most recently, \cite{hu2025generalized} showed that calibrated multiaccuracy can be used to extend and improve the efficiency of prior regularity-based characterizations of computational notions of entropy \cite{vadhan2012pseudoentropy, vadhan2013uniform, zheng2014thesis, casacuberta2024complexity}.

\paragraph{Property Testing}

The study of property testing writ large was initiated by \cite{rubinfeld1996property,goldreich1996property}. Within this field, several works have attempted to characterize which properties are testable under various constraints on the tester's power. For example, in the graph context, \cite{alon2006regularity} proved, roughly speaking, that a property of dense graphs is testable if and only if it can be determined from a Szemer\'{e}di regular partition of the graph. In some sense, this result was a capstone to a large body of work devoted to understanding the testability of graph properties, including \emph{monotone} properties \cite{alon2005monotone}, \emph{hereditary} properties \cite{alon2005hereditary}, and more.

In the present paper, we focus not only on properties of dense graphs, but also on properties of a certain class of dense distributions defined by Boolean functions. As already discussed, in this context, the closest related work to ours is the characterization of constant-sample testability in terms of constant-part symmetry, due to \cite{blais2019testable}. Even earlier, \cite{kaufman2008invariance,sudan2010invariance} pioneered the idea that the symmetries of a property play a central role in understanding its testability. Of course, one difference is that our work concerns the \emph{computational} complexity of property testing, which has recently received renewed interest \cite{pinto2025testing}.

\paragraph{Private Graph Algorithms} There have been several works designing differentially private algorithms for approximating the \emph{cut function} of a graph $G = (V, E)$ \cite{gupta2012iterative,arora2019sparsification,marek2020synthetic,liu2024optimal,zou2025almost,amand2025breaking}. This function takes as input two vertex subsets $S,T \subseteq V$, and outputs their edge count $e_G(S, T)$. Our algorithm for privately compute regularity templates has similar (but, technically speaking, incomparable) goals and guarantees. While state-of-the-art algorithms for privately approximating the cut function generally take $\mathrm{poly}(n)$ or $\exp(n)$ time and achieve $O(n^c)$ error for $c < 2$, our algorithm for compute regularity templates allows one to approximate cut sizes (without location information) with $\eps n^2$ by inspecting only a constant $O_\eps(1)$ number of vertices of the graph. Phrased differently, our algorithm computes a coarser approximation of the graph, but does so much more efficiently---indeed, more efficiently than one might expect to be possible.

\subsection{Paper Organization}

In \cref{sec:preliminaries}, we cover preliminaries on property testing, regularity, and supersimulators. In \cref{sec:symmetry-thm}, we prove our main results on efficient property testing, \cref{thm:main-hard,thm:main-easy}, as well as our result on consistency counting, \cref{thm:cc}, and our extension to dense distribution testing, \cref{thm:distr-hard}. In \cref{sec:dp-graphs}, we present our results on node-level private graph testing, \cref{thm:dp-partition,thm:dp-graph-testing}, as well as the related function testing results, \cref{thm:templates-hard,thm:templates-easy}. In \cref{sec:mpv}, we prove our results on the computational indistinguishability of product distributions: \cref{thm:mpv-cma-succinct,thm:mpv-super-succinct}. In \cref{sec:supersimulators}, we present the constructions of supersimulators on which we rely. These imply the version stated in \cref{thm:super-succinct} as a special case.
 
\section{Preliminaries}
\label[section]{sec:preliminaries}

In this paper, $\cX$ denotes an arbitrary finite set, $\Delta(\cX)$ denotes the set of probability distributions on $\cX$, and $\{\cX \to \cY\}$ denotes the set of functions from $\cX$ to a set $\cY$. Given $\cF \subseteq \{\cX \to \R\}$ and $c \in \R$, let $c \cdot \cF$ denote the set of functions $c \cdot f$ for $f \in \cF$. Let $-\cF = (-1) \cdot \cF$ and let $\pm \cF = \cF \cup -\cF$. Let $\cB(p)$ denote the Bernoulli distribution with parameter $p \in [0, 1]$. Finally, let $[t]_a^b$ denote the projection of $t \in \R$ onto the interval $[a, b]$.

\paragraph{Property Testing}
A \emph{property} $\cP$ is a set of Boolean functions $f: \{0, 1\}^n \to \{0, 1\}$. We say that \emph{$f$ has the property $\cP$} if $f \in \cP$. The \emph{distance} between two functions is the fraction of inputs on which they disagree. We write $f \in \cP_\eps$ if $f$ is $\eps$-close to some $g \in \cP$. While property testing can be studied in either a query-based or sample-based access model, we focus on the latter perspective, which was introduced by \cite{goldreich1996property}:

\begin{definition}[Sample-Testable Property]
\label[definition]{def:sample-testability}
    Let $\cP$ be a property of Boolean functions. We say that $\cP$ is \emph{sample-testable with proximity parameter $\eps > 0$} if there is a randomized circuit $T$ of size $s$ that receives as input $m$ independent samples $x_i \sim \{0, 1\}^n$ and their labels $y_i = f(x_i)$, always outputs either \textsc{Accept} or \textsc{Reject}, and meets the following two requirements:
    \begin{itemize}
        \item If $f$ has the property $\cP$, then $T$ outputs \textsc{Accept} with probability at least $2/3$.
        \item If $f$ is $\eps$-far from having $\cP$, then $T$ outputs \textsc{Reject} with probability at least $2/3$.
    \end{itemize}
    In both conditions, the probability is computed over randomness in the sample and internal to $T$.
\end{definition}

More formally, we will model the tester as a deterministic function $T : (\cX \times \{0, 1\})^m \times \{0, 1\}^\ell \to \{0, 1\}$ that receives as input $m$ labeled samples $(x_i, y_i) \in \cX \times \{0, 1\}$ and a uniform $\ell$-bit random seed $r \in \{0, 1\}^\ell$. Often, we will write  $x \in \cX^m$ and $y \in \{0, 1\}^m$ and use the abbreviation
\[
    T(x, y, r) = T\bigl((x_1, y_1), \ldots, (x_m, y_m), r\bigr).
\]
We say that $T$ ``accepts'' when it outputs $1$, and ``rejects'' when it outputs $0$. Sometimes, it will be convenient to work directly with the expected value of $T$ over its internal randomness, or, equivalently, the mean function $\bar{T} : (\cX \times \{0, 1\})^m \to [0, 1]$ defined by the formula
\[
    \bar{T}(x, y) = 2^{-\ell}\sum_{r \in \{0, 1\}^\ell} T(x, y, r).
\]
When $\cX = \{0, 1\}^n$, we will often discuss the circuit size of the tester, by which we mean $T$, not $\bar{T}$.

\paragraph{Structured Symmetry} A property $\cP$ of Boolean functions is \emph{$k$-part symmetric} if there is a partition of $\{0, 1\}^n$ into disjoint parts $S_1, \ldots, S_k$ such that $\cP$ is invariant under permutations of the domain within each part. Equivalently, whether or not a function $f$ has the property $\cP$ can be completely determined from the $k$ scalar \emph{densities} $\E[f(x) \bm{1}[x \in S_j]]$, where the expectation is computed over a random input $x \sim \{0, 1\}^n$.

In general, such parts need not have any special structure, and may be very complex. In contrast, we say that a part $S_j$ has \emph{computational complexity} at most $s$ if there is a circuit of size $s$ that decides whether or not a given input belongs to $S_j$. We say that a partition $\cP$ has computational complexity at most $s$ if there is a circuit of size at most $s$ computing its \emph{classification function}, which outputs the index of the unique part $S_j$ containing a given input $x$.

\paragraph{Complexity-Theoretic Regularity}
First, we define regularity and indistinguishability.

\begin{definition}[Regularity and Indistinguishability]
\label[definition]{def:indistinguishability}
    Given a family $\cF \subseteq \{\cX \to [-1, 1]\}$, an error parameter $\delta > 0$, and two functions $g, h : \cX \to [0, 1]$, we say that $h$ is a \emph{$(\cF, \delta)$-regular simulator for $g$ under $\cD$} if for all distinguisher functions $f \in \cF$,
    \[
        \Bigl\lvert\scalebox{1.2}{$\mathbb{E}$}_{x \sim \cD}\bigl[f(x)\bigl(g(x)-h(x)\bigr)\bigr]\Bigr\rvert \le \delta.
    \]
    Equivalently, we say that $g$ and $h$ are \emph{$(\cF, \delta)$-indistinguishable} (when clear from context, we will omit the phrase ``under $\cD$''). When $\cD$ is the uniform distribution over $\cX = \{0, 1\}^n$ and $\cF$ is the collection of all Boolean circuits of size at most $s$, we say that $h$ is an \emph{$(s, \delta)$-regular} simulator for $g$, or that the functions are \emph{$(s, \delta)$-indistinguishable}.
\end{definition}

Given any target function $g : \cX \to [0, 1]$, there exists a trivial $(\cF, \delta)$-regular simulator for $g$, namely $h = g$. The \emph{complexity-theoretic regularity lemma} guarantees the existence of a much better simulator, whose complexity does not scale with $\cX$ or $g$, but rather depends only on $\cF$ and $\delta$. We will state the lemma in terms of the distinguisher family's \emph{structured sums} $\cS_{k,\delta}(\cF)$.

\begin{definition}[Structured Sums]
\label[definition]{def:struct-sum}
    Given a distinguisher family $\cF \subseteq \{\cX \to [-1, 1]\}$ and an initialization function $h_0 : \cX \to [0, 1]$, the set $\cS_{k,\delta}(\cF, h_0)$ comprises all functions of the form \[h(x) = \Bigl[h_0(x) + \delta\cdot \bigl(f_1(x) + \cdots + f_k(x)\bigr)\Bigr]_0^1\] for some $f_1, \ldots, f_k \in \pm \cF$.\footnote{Recall that $[\cdot]_0^1$ projects onto $[0, 1]$.} Let $\cS_{< k, \delta}(\cF; h_0) = \bigcup_{j<k} \cS_{j,\delta}(\cF, h_0)$.
\end{definition}

We will often use the initialization $h_0 = 0$. When this is clear from context, we will simply write $\cS_{k,\delta}(\cF)$ and $\cS_{<k,\delta}(\cF)$ in place of $\cS_{k,\delta}(\cF, 0)$ and $\cS_{<k,\delta}(\cF, 0)$.

\begin{lemma}[Complexity-Theoretic Regularity \cite{trevisan2009regularity}]
\label[lemma]{thm:ttv}
    Fix $\cD \in \Delta(\cX)$, $\cF \subseteq \{\cX \to [-1, 1]\}$, $h_0 : \cX \to [0, 1]$ and $\delta > 0$. Every $g : \cX \to [0, 1]$ has an $(\cF, \delta)$-regular simulator $h \in \cS_{<(2/\delta^2),(\delta/2)}(\cF, h_0)$.
\end{lemma}

In the introduction, we had focused on the special case where $\cD$ is the uniform distribution on $\cX = \{0, 1\}^n$ and $\cF$ contains all circuits of size at most $s$, but this need not be the case. We had also considered randomized distinguishers $A(x, y)$ that take as a second input a binary label $y$ sampled from either $\cB(g(x))$ or $\cB(h(x))$, but here we consider deterministic distinguishers $f(x)$. Of course, the latter perspective subsumes the former. Given $A$, simply set $f_A(x) = \E[A(x, 1) - A(x, 0)]$, where the expectation is taken over the randomness of $A$. Then, the distinguishing advantage of $A$ is precisely $\abs{\E_{x \sim \cD}[f_A(x)(g(x) - h(x))]}$.

The complexity-theoretic regularity lemma, as stated in \cref{thm:ttv}, may produce a simulator $h$ that is more complex than the functions in $\cF$, but for our purposes, we will need $h$ to be a \emph{supersimulator}, which fools distinguishers more complex than itself. We formalize this in terms of a \emph{growth function} $\cG$ that takes as input any function $h : \cX \to [0, 1]$ and outputs the distinguisher family $\cG(h) \subseteq \{\cX \to [-1, 1]\}$ that we would like $h$ to fool. Now, the notion of structured sums from \cref{def:struct-sum} must be adjusted accordingly.

\begin{definition}[Structured Sums]
    Fix $\delta > 0$. Given an initialization $h_0 : \cX \to [0, 1]$ and distinguisher functions $f_1, f_2, \ldots : \cX \to [-1, 1]$, let $f_{1:0} = h_0$ and for each $j \in [k]$, let
    \[
        f_{1:j}(x) = \Bigl[h_0(x) + \delta\cdot \bigl(f_1(x) + \cdots + f_j(x)\bigr)\Bigr]_0^1.
    \]
    If $\cG$ is a growth function, then $\cS_{k,\delta}(\cG, h_0)$ is the set of functions of the form of $f_{1:k}$ for some $f_1, f_2, \ldots$ such that each $f_j$ belongs to the previous family $\pm \cG(f_{1:j-1})$. Let $\cS_{< k,\delta}(\cG, h_0) = \bigcup_{j<k} \cS_{j,\delta}(\cG, h_0)$.
\end{definition}

Once again, we omit the second argument $h_0$ when it is clear from context that $h_0 = 0$.

\begin{restatable}[Supersimulators]{lemma}{supersimulators}
\label[lemma]{thm:super}
    Fix $\cD \in \Delta(\cX)$, a growth function $\cG$, an initialization $h_0 : \cX \to [0, 1]$, and $\delta > 0$. Every $g : \cX \to [0, 1]$ has a $(\cG(h), \delta)$-regular simulator $h \in \cS_{<(2/\delta^2),(\delta/2)}(\cG, h_0)$.
\end{restatable}

Note that \cref{thm:ttv} is a special case of \cref{thm:super} corresponding to a constant growth function that always outputs $\cF$. While supersimulators appeared implicitly in \cite{dwork2021outcome}, the statement of \cref{thm:super} differs slightly from the version in \cite{dwork2021outcome}. For this reason, we give a short proof of \cref{thm:super} in \cref{sec:supersimulators}, along with our alternate construction of supersimulators using an iteration technique from the graph regularity literature.

\paragraph{Multicalibration}

A simulator (a.k.a. predictor) $h$ is said to be \emph{calibrated} if \(\E_{x \sim \cD}[g(x) | h(x)] = h(x)\). Although the simulator provided by \cref{thm:ttv} is not necessarily perfectly calibrated, one can ensure that it is \emph{approximately calibrated} at the cost of a small increase in the complexity of $h$. Specifically, following \cite{gopalan2022low}, we say that $h$ is $\gamma$-calibrated if for all functions $w : [0, 1] \to [0, 1]$,
\[
    \Bigl\lvert\E_{x \sim \cD}[w(h(x))(g(x) - h(x)]\Bigr\rvert \le \gamma.
\]
The following known result requires only slight modifications to the proof of \cref{thm:ttv}.  We will state it in terms of the class $\cF_{(s_1, s_2)}$ of functions $h$ with the following property: there exist functions $f_1, \ldots, f_k \in \cF$ with $k \le s_1$ and a Boolean circuit of size at most $s_2$ that computes the output of $h(x)$ in binary given only $f_1(x), \ldots, f_k(x)$ as input.

\begin{lemma}
\label[lemma]{thm:ttv-calibrated}
    For all $\cD \in \Delta(\cX)$, $\cF \subseteq \{\cX \to [-1, 1]\}$, $g : \cX \to [0, 1]$, and $0 < \gamma \le \eps \le 1$, there exists an $(\cF, \eps)$-regular and $\gamma$-calibrated simulator in $\cF_{(O(1/\eps^{2}), \, \tilde{O}(1/\gamma^3))}$.
\end{lemma}

We have chosen to work with a fairly strong definition of approximate calibration. Various relaxations exist in the literature that would lead to an even milder dependence on $\gamma$ in the complexity of $h$ \cite{gopalan2022low}. More stringent than calibration is \emph{multicalibration} \cite{hkrr, kearns2018gerrymandering}, which is essentially a per-level-set regularity requirement on the simulator. We say that $h$ is $(\cF, \eps)$-multicalibrated if
\[
    \Pr_x\biggl[\max_{f \in \cF} \Bigl\lvert\E\bigl[f(x)(g(x) - h(x)) \,\big|\,h(x)\bigr]\Bigr\rvert \le \eps\biggr] \ge 1 - \eps.
\]
(Several variations of the definition of multicalibration exist in the literature, and the one we have stated is often referred to as \emph{swap} multicalibration or \emph{strict} multicalibration.) To emphasize the relationship between $(\cF, \eps)$-regularity and $(\cF, \eps)$-multicalibration, we will sometimes follow the now-standard convention of referring to the former as \emph{$(\cF, \eps)$-multiaccuracy}, where ``accuracy'' is short for accuracy in expectation.

One can check that $(\cF, \eps)$-multicalibration implies both $(\cF, O(\eps))$-multiaccuracy and, if $\cF$ contains the constant $0$ and $1$ functions, $O(\eps)$-calibration. Moreover, $(\cF, \eps)$-multicalibration can be achieved at the cost of a moderate increase in the complexity of $h$ relative to $\cF$:

\begin{lemma}[\cite{hkrr}]
\label[lemma]{thm:multicalibration}
    For all $\cD \in \Delta(\cX)$, $\cF \subseteq \{\cX \to [-1, 1]\}$, $g : \cX \to [0, 1]$, and $0 < \eps \le 1$, there exists an $(\cF, \eps)$-multicalibrated simulator in $\cF_{(O(1/\eps^{4}), \, \tilde{O}(1/\eps^4))}$.
\end{lemma}

We emphasize that \cref{thm:multicalibration} will not be used directly in any of our proofs, but it is nevertheless a useful point of reference. We also remark all of these concepts have been studied in the so-called \emph{multiclass} case, corresponding to $g : \cX \to \Delta([k])$. In this case, which we do not study in this paper, the complexity of known multicalibration constructions are dramatically worse, degrading exponentially with $k$, motivating the study of whether more lightweight regularity notions suffice for downstream applications.

\paragraph{Graph Regularity Iteration} Given a simple graph $G = (V, E)$, let $g = \bm{1}_E: V \times V \to \{0,1\}$ denote its edge indicator function, and let $\cF = \{\bm{1}_{S \times T} : S, T \subseteq V\}$ contain the indicator functions for all rectangles $S \times T \subseteq V \times V$, sometimes called \emph{cuts}. Famous and well-studied notions of graph regularity, such as Szemer\'{e}di regularity \cite{szemeredi1975regular} and Frieze-Kannan regularity \cite{frieze1996regularity, frieze1999approximation}, are closely related to multicalibration and multiaccuracy with respect to $\cF$, respectively \cite{trevisan2009regularity, skorski2017crypto, dwork2023pseudorandomness}. This connection suggests that techniques developed in the graph regularity literature may have counterparts in our complexity-theoretic setting.

We confirm this intuition in \cref{sec:supersimulators}, where we adapt the technique of \emph{iteration}, which constructs a sequence of regular partitions $\cP_1, \cP_2, \ldots$ of the same graph $G$. In this method, each $\cP_{i+1}$ is a refinement of $\cP_i$, and each $\cP_{i+1}$ is extremely regular relative to the complexity of $\cP_i$ (more formally, $\cP_{i+1}$ achieves regularity with an error parameter $\eps_{\abs{\cP_i}}$ that decays with the number of parts of $\cP_i$). With some care, one can ensure that there exist consecutive partitions $\cP = \cP_{i^\star}$ and $\cQ = \cP_{i^\star + 1}$, where $i^\star \le O(1/\eps^2)$, such that the $\cP$ and $\cQ$ are distance at most $\eps$ in an appropriately defined metric. Depending on the base regularity notion chosen to instantiate this technique, the resulting pair $(\cP, \cQ)$ may have substantially stronger regularity properties than $\cP$ alone, allowing one to derive, for example, Szemer\'{e}di regularity from the weaker notion of Frieze-Kannan regularity, as well as improved bounds for various combinatorial applications. For more detail about this technique in the graph setting, we refer the reader to \cite{zhao2023gtac}, as well as \cite{rodl2010regularity, fox2011new, conlon2012bounds}.

\paragraph{Graph Regularity Templates} In \cref{sec:intro-privacy}, we defined \emph{regular partitions}, \emph{regularity templates}, and \emph{regular-reducibility}. In order to prove our results regarding node-level private graph property testing, we will need a few more technical definitions and results regarding these concepts. The first lemma that we need measures the number of edge edits required to improve regularity by a desired amount $\delta$.

\begin{lemma}[\cite{alon2006regularity}]
\label[lemma]{thm:tweak}
    Let $R = (\gamma, k, \eta, \bar{R})$ be a regularity template and let $(V_1, \ldots, V_k)$ be an equipartition of a graph $G$ such that $\abs{d(V_i, V_j) - \eta_{ij}} \le \delta$ for all pairs $(i, j) \in [k]^2$ and $(V_i, V_j)$ is $(\gamma + \delta)$-regular for all pairs $(i, j) \notin \bar{R}$. Then $G$ is $O(\delta/\gamma^2)$-close to satisfying $R$.
\end{lemma}

The next definition that we need is a notion of similarity between two vertex equipartitions (not necessarily of the same underlying graph).

\begin{definition}[\cite{alon2006regularity}]
\label[definition]{def:reg-sim}
    An equipartition $\cU = \{U_1, \ldots, U_k\}$ of a graph $G$ is \emph{$\delta$-similar} to a $\gamma$-regular equipartition $\cV = \{V_1, \ldots, V_k\}$ of a graph $H$ if $\abs{d_G(U_i, U_j) - d_H(V_i, V_j)} \le \delta$ for all pairs $(i, j) \in [k]^2$ and if $(U_i, U_j)$ is $(\gamma+\delta)$-regular whenever $(V_i, V_j)$ is $\gamma$-regular.
\end{definition}

The next lemma states that one can identify regular partitions of a graph by analyzing a constant-size subsample of the vertices. This lemma is intuitive but certainly not trivial, since it is far from obvious that so few samples suffice for this purpose.

\begin{lemma}[\cite{alon2006regularity}]
\label[lemma]{thm:reg-samp}
    For all $r \in \N$ and $\delta > 0$, there exists $q = O_{r,\delta}(1)$ such that if $\gamma \ge \delta$ and $k \le r$, then the following property holds with probability at least $2/3$ over a uniformly random sample $Q$ of $q$ vertices from a graph $G$: for every size-$k$ $\gamma$-regular equipartition $\cV$ of $G$, there exists a size-$k$ equipartition $\cU$ of $G[Q]$ that is $\delta$-similar to $\cV$, and for every size-$k$ $\gamma$-regular equipartition $\cU$ of $G[Q]$, there exists a size-$k$ equipartition $\cV$ of $G$ that is $\delta$-similar to $\cU$.
\end{lemma}

The next lemma asserts that any graph property that is testable from from a constant number of edge queries is also testable by inspecting the subgraph induced by a constant number of \emph{randomly} sampled vertices. The intuition behind the lemma is that a graph property is by definition invariant under vertex relabeling, so it should not matter which part of the graph we inspect.

\begin{lemma}[\cite{alon1999efficient,goldreich2001three}]
\label[lemma]{thm:graph-canonical}
    If property $\cP$ of $n$-vertex graphs is testable for all $\eps > 0$ with proximity $\eps$ and $q/2 = O_\eps(1)$ edge queries, then $\cP$ is testable for all $\eps > 0$ with a tester that first selects a uniformly random subset $Q$ of $q$ vertices from $G$, then queries all the pairs in $Q$, and finally outputs \textsc{Accept} or \textsc{Reject} in a deterministic manner that depends only on the isomorphism class of the induced subgraph $G[Q]$.
\end{lemma}

The final result we need is the equivalence between tolerant and non-tolerant property testing in the dense graph model.

\begin{lemma}[\cite{fischer2005estimation}]
\label[lemma]{thm:graph-tol}
    If for all $\eps > 0$, a graph property $\cP$ is testable with proximity $\eps$ using $O_\eps(1)$ edge queries, then for all $\eps > 0$ and $\delta \ge 0$, the graph property $\cP_{\delta}$ is testable with proximity $\eps$ using $O_{\eps,\delta}(1)$ edge queries (i.e. $\cP$ is tolerantly testable with proximity parameters $\delta$ and $\delta + \eps$).
\end{lemma}

\paragraph{Differential Privacy}

For our results on private graph regularity, we will need the definition of DP from \cite{dwork2006calibrating}, specialized to \emph{node-adjacent} graphs, as in \cite{blocki2013social,chen2013recursive,kasiviswanathan2013node}, the utility guarantees for the classic \emph{exponential mechanism} \cite{mcsherry2007mechanism}, and the powerful technique of \emph{privacy amplification by subsampling} (see \cite{balle2018amplification,steinke2022composition} for an overview of this technique).

\begin{definition}[Node-Level Differential Privacy \cite{blocki2013social,chen2013recursive,kasiviswanathan2013node}]
    Two graphs $G, G'$ are \emph{node-adjacent} if $G'$ can be obtained by selecting a single vertex $v$ and arbitrarily adding or removing edges from $G$ that are incident to $v$. A randomized algorithm $\cA$ that takes as input a graph and outputs an element of a set $\cS$ satisfies \emph{node-level $\alpha$-differential privacy ($\alpha$-DP)} if for all node-adjacent graphs $G, G'$ and all possible outputs $s \in \cS$,
    \[
        \Pr[\cA(G) = s] \le e^\alpha \Pr[\cA(G') = s],
    \]
    where the probabilities are computed over \emph{only} the randomness internal to the algorithm $\cA$.
\end{definition}

\begin{lemma}[Exponential Mechanism \cite{mcsherry2007mechanism}]
\label[lemma]{thm:expmech}
    Let $\ell_1, \ell_2$ be two real-valued functions on a finite domain $\cR$ that differ pointwise by at most $\Delta$, and let $\cD_i$ be the distribution over $\cR$ proportional to $\exp(-(\alpha/2\Delta) \cdot \ell_i(R))$. Then, $\cD_1(R) \le e^\alpha \cdot \cD_2(R)$ for each $R \in \cR$, and for all $\beta > 0$,
    \[
        \Pr_{R \sim \cD_1} \Bigl[\ell_1(R) \le \min_{R^* \in \cR} \ell_1(R^*) + \frac{2\Delta}{\alpha}\log\Bigl(\frac{\abs{\cR}}{\beta}\Bigr)\Bigr] \ge 1-\beta.
    \]
\end{lemma}

\begin{lemma}[Privacy Amplification by Subsampling, Variant of Theorem 29 of \cite{steinke2022composition}]
\label[lemma]{thm:amp}
    Let $Q \subseteq [n]$ be a random subset. Assume that, for all $i \in [n]$, we can define a random subset $Q_{-i} \subseteq [n] \setminus \{i\}$ distributed jointly with $Q$ such that the following two conditions hold for all $i \in [n]$:
    \begin{itemize}
        \item For all $n$-vertex graphs $G$, the induced subgraphs $G[Q]$ and $G[Q_{-i}]$ are node-adjacent with probability $1$ over the randomness of $Q$ and $Q_{-i}$,
        \item The marginal distribution of $Q_{-i}$ conditioned on $i \in Q$ is equal to the marginal distribution of $Q$ conditioned on $i \notin Q$.
    \end{itemize}
    Let $M$ be an $\alpha$-DP mechanism that takes as input an $n$-vertex graph. Define the mechanism $M^Q(G) = M(G[Q])$. Let $p = \max_{i \in [n]} \Pr_Q[i \in Q]$. Then $M^Q$ is $\alpha'$-DP for $\alpha' = \log(1 + p(e^\alpha-1))$.
\end{lemma}

Whenever we use \cref{thm:amp}, we will let $Q$ be a uniformly random subset of $[n]$ of size exactly $q$, for some desired size parameter $q \in [n]$. We will let $Q_{-i}$ be the result of replacing the element $i$ (if present in $Q$) with a uniformly random element in $[n] \setminus Q$, as described in \cite{steinke2022composition}.

\section{Finding Structured Symmetry}
\label[section]{sec:symmetry-thm}

In this section, we prove \cref{thm:main-hard,thm:main-easy}, which comprise our main equivalence between efficient testability and structured symmetry. As discussed in the introduction, this result strengthens the main theorem of \cite{blais2019testable} by capturing the relationship between the tester's complexity and the structure of the property's invariant sets. For convenience, we recall the statements of these two theorems below. The first, \cref{thm:main-hard}, is the harder of the two, showing that efficient testing implies structured symmetry.

\mainhard*

Like the main theorem of \cite{blais2019testable}, our \cref{thm:main-hard} is most interesting for properties testable from at most $O(\log n)$ samples, where $2^n$ is the domain size. We will extend our results to properties testable with larger sample sizes in subsequent sections. The second theorem, \cref{thm:main-easy}, establishes the converse to \cref{thm:main-hard} by showing that structured symmetry implies efficient testing. 

\maineasy*

To see why \cref{thm:main-easy} is indeed a converse to \cref{thm:main-hard}, note that if $\cP \subseteq \cQ \subseteq \cP_{\eps_1}$ for some $\eps_1 > 0$, then $\cQ_{\eps_2} \subseteq \cP_{\eps_1 + \eps_2}$ for any $\eps_2 > 0$. Consequently, any tester for $\cQ$ with proximity $\eps_2$ is also a tester for $\cP$ with proximity $\eps_1 + \eps_2$. Thus, when combined, \cref{thm:main-hard,thm:main-easy} imply that a property of Boolean functions is constant-sample testable with a small circuit if and only if it is close to having constant-part structured symmetry.

While the proof of \cref{thm:main-easy} is quite simple, the proof of \cref{thm:main-hard} is not, requiring two key lemmas. We call them the \emph{oracle simulation} lemma and the \emph{tester simulation} lemma, and we prove them in \cref{sec:oracle-sim,sec:tester-sim}, respectively. While both lemmas involve regular simulators, they differ in terms of which objects play the roles of the target function and distinguisher family. For the oracle simulation lemma, the target function is the one we wish to test for $\cP$, which defines the sample-generating oracle, and the distinguishers are derived from the tester. For the tester simulation lemma, the tester plays the role of the target function, and the distinguisher family comprises low-complexity approximations to the functions to be tested for $\cP$. In \cref{sec:constant-sample-proofs}, we combine these lemmas to prove \cref{thm:main-hard}, and we also prove \cref{thm:main-easy}.

\subsection{Simulating the Oracle}
\label[section]{sec:oracle-sim}

In this section, we consider the probability that an $m$-sample tester $T$ accepts or rejects inputs generated by any particular sample-generating oracle, which is defined by a distribution $\cD \in \Delta(\cX)$ and a target function $f : \cX \to [0, 1]$. We show that this probability remains roughly the same upon replacing $f$ with any $(\cR(T), \delta)$-regular simulator $\tilde{f}$, where $\cR(T)$ is the simple distinguisher family described in \cref{def:restrictions} below. Specifically, $\cR(T)$ is a family of \emph{one-way restrictions} of $T$, obtained by hard-wiring fixed values for all but one of its inputs. Later, when we take $T$ to be a size-$s$ circuit, it will be clear that every distinguisher of this form is also a size-$s$ circuit.

The main result of this section, which we state formally in \cref{thm:oracle-sim}, is significant because of the way it facilitates the construction of additional property testers. To illustrate this informally, suppose that a property is testable by a simple circuit $T$ and that we would like to argue that some other function $\widetilde{T}$ is \emph{also} a valid tester. One way to do this would be to argue that $\widetilde{T}$ behaves similarly to $T$ for all possible sample-generating oracles, but this may be difficult. \Cref{thm:oracle-sim} shows it suffices to check that $\widetilde{T}$ behaves similarly to $T$ for sample-generating oracles defined by \emph{low-complexity} target functions $\tilde{f}$. We note, however, that the result of this section will be stated without the language of property testing, as it depends only on the function $T$.

Before proceeding with the statement and proof, we first define the relevant distinguisher class. The definition is precisely the one needed to enable a certain hybrid argument that we plan to carry out. We state the definition in a general manner so that it is applicable to both the actual tester $T : (\cX \times \{0, 1\})^m \times \{0, 1\}^\ell \to \{0, 1\}$ and its mean function $\bar{T} : (\cX \times \{0, 1\})^m \to [0, 1]$, which averages over the choice of the random seed $r\in\{0,1\}^\ell$.

\begin{definition}[Restriction Distinguishers]
\label[definition]{def:restrictions}
    Given a function $T : (\cX \times \{0, 1\})^m \times \{0, 1\}^\ell \to [0, 1]$, consider the function $T_{x_{\neq i}, y, r} : \cX \to [0, 1]$ which hard-wires all inputs to $T$ except for $x_i$:
    \[
        T_{x_{\neq i}, y, r}(x) = T\bigl((x_1, y_1), \ldots, (x_{i-1}, y_{i-1}), (x, y_i), (x_{i+1}, y_{i+1}), \ldots, (x_m, y_m), r \bigr).
    \]
    We define $\cR(T)$ to be the set of these \emph{one-way restriction} functions $T_{x_{\neq i}, y, r}$ for all indices $i \in [m]$, sequences $x_{\neq i}$ comprising values $x_j \in \cX$ for each $j \neq i$, labels $y \in \{0, 1\}^m$, and seeds $r \in \{0, 1\}^\ell$.
\end{definition}

In the special case that $T$ is Boolean-valued, each function $T_{x_{\neq i}, y, r}$ is also Boolean-valued. With this definition in hand, we are now ready to state the main result of this section:

\begin{lemma}[Oracle Simulation]
\label[lemma]{thm:oracle-sim}
    Fix a function $T : (\cX \times \{0, 1\})^m \times \{0, 1\}^\ell \to [0, 1]$, a distribution $\cD \in \Delta(\cX)$, and $\delta > 0$. If $\tilde{f}$ is an $(\cR(T), \delta)$-regular simulator for  $f : \cX \to [0, 1]$, then
    \[
        \Bigl\lvert \E\bigl[T(x, y, r)\bigr] - \E\bigl[T(x, \tilde{y}, r)\bigr] \Bigr\rvert \le 2m\delta,
    \]
    where $x_i \sim \cD$ and $y_i|x_i \sim \cB(f(x_i))$ and $\tilde{y}_i|x_i \sim \cB(\tilde{f}(x_i))$ for each $i \in [m]$ and $r \sim \{0, 1\}^\ell$.
\end{lemma}

A remark is in order regarding the random variables $y_i$ and $\tilde{y}_i$. In the language of \emph{Outcome Indistinguishability (OI)} \cite{dwork2021outcome} from the algorithmic fairness literature, $y_i$ and $\tilde{y}_i$ correspond precisely to \emph{real} and \emph{modeled} outcomes of an individual represented by the features $x_i$. From this point of view, $f$ defines the ground-truth or Bayes optimal conditional probability distribution, and $\tilde{f}$ corresponds to a predictor of $f$ satisfying \emph{no-access OI} with respect to $\cR(T)$.

The proof of \cref{thm:oracle-sim} involves two simple components. The first component is a hybrid argument, similar to one used in recent work on multicalibration-based characterizations of the indistinguishability of product distributions \cite{marcussen2024characterizing}. The other component is a standard transformation between distinguishers that receive a labeled or unlabeled input.

\begin{proof}
    First, for each integer $1 \le i \le m$, define
    \[
        T_i(x,y) = T\bigl((x_1, y_1), \ldots, (x_{i-1}, y_{i-1}), (x, y), (x_{i+1}, \tilde{y}_{i+1}), \ldots, (x_m, \tilde{y}_m), r\bigr).
    \]
    Observe that the function $T_i : \cX \times \{0, 1\} \to \{0, 1\}$ depends on the values of $x_j$ for all $j \neq i$, the values of $y_j$ for all $j < i$, the values of $\tilde{y}_j$ for all $j > i$, and the tester's random seed $r$. Next, for each $0 \le i \le m$, let the (deterministic) scalar $a_i$ denote the expected output of the tester $T$ under the \emph{$i$\textsuperscript{th} hybrid distribution}, in which the first $i$ labels are real and the remaining $m-i$ labels are modeled. More formally, for $1 \le i \le m-1$, we define
    \[
        a_i = \E\bigl[T_i(x_i, y_i)\bigr] = \E\bigl[T_{i+1}(x_{i+1}, \tilde{y}_{i+1})\bigr].
    \]
    When $i = m$, we define $a_m$ via the first of these two expressions, and when $i = 0$, we define $a_0$ via the second of these two expressions. Ultimately, our goal is to bound $\abs{a_m - a_0} \le 2m\delta$. By the triangle inequality, it suffices to show that $\abs{a_i - a_{i-1}}\le 2\delta$ for each index $i \in [m]$. For this, we condition on $x_j$ for all $j \neq i$ and $y_j$ for all $j < i$ and $\tilde{y}_j$ for all $j > i$ (i.e. everything except for $x_i$, $y_i$ and $\tilde{y}_i$):
    \[
        \abs{a_i - a_{i-1}} \le \E\Bigl\lvert\E\bigl[T_i(x_i, y_i) - T_i(x_i, \tilde{y}_i)\,\big|\,x_{\neq i}, y_{<i}, \tilde{y}_{>i},r\bigr]\Bigr\rvert.
    \]
    Note that the inner expectation is \emph{only} over the randomness in the $i$\textsuperscript{th} coordinate (i.e. $x_i$, $y_i$, and $\tilde{y}_i$), since we have conditioned on everything else. The outer expectation is over the randomness of these other variables, namely $x_{\neq i}$, $y_{<i}$, $\tilde{y}_{>i}$ and $r$.
    
    At this point, we have computed an upper bound on $\abs{a_i - a_{i-1}}$ in terms of the distinguishing advantage of a circuit $T_i$ derived from hard-wiring all but two of the inputs to $T$. However, we are not quite done until we have hard-wired all but one input, since this is the form required by $\cR(T)$. To this end, we apply a standard transformation to $T_i$. First, since $y_i|x_i \sim \cB(f(x_i))$, we rewrite the conditional expectation of $T_i(x_i, y_i)$ in terms of $f(x_i)$ instead of $y_i$:
    \[
        \E\bigl[T_i(x_i, y_i)\bigr] = \E\bigl[T_i(x_i, 0) + (T_i(x_i,1) - T_i(x_i,0))f(x_i)\bigr].
    \]
    Since $\tilde{y}_i|x_i \sim \cB(\tilde{f}(x_i))$, a similar formula holds for $\tilde{y}_i$ with $\tilde{f}$ in place of $f$:
    \[
        \E\bigl[T_i(x_i, \tilde{y}_i)\bigr] = \E\bigl[T_i(x_i, 0) + (T_i(x_i,1) - T_i(x_i,0))\tilde{f}(x_i)\bigr].
    \]
    We observe that these two formulas remain true even if we condition on the values of $x_i$, $\,x_{\neq i}$, $y_{<i}$, $\tilde{y}_{>i}$, and $r$. Therefore, subtracting the two equations yields
    \[
        \abs{a_i - a_{i-1}} \le \E\Bigl\lvert\E\bigl[\bigl(T_i(x_i, 1) - T_i(x_i, 0)\bigr)\bigl(f(x_i)-\tilde{f}(x_i)\bigr)\,\big|\,x_{\neq i}, y_{<i}, \tilde{y}_{>i},r\bigr]\Bigr\rvert.
    \]
    The functions $T_i(x, 1)$ and $T_i(x, 0)$ clearly belong to $\cR(T)$. 
    Since $\tilde{f}$ is an $(\cR(T), \delta)$-regular simulator for $f$, we have $\abs{a_i - a_{i-1}} \le 2\delta$, so we conclude that $\abs{a_m - a_0} \le 2m\delta$.
\end{proof}

\subsection{Simulating the Tester}
\label[section]{sec:tester-sim}

In this section, we show that replacing $T$ with a suitable \emph{supersimulator} $\widetilde{T}$ only slightly affects our chance of accepting or rejecting any sufficiently low-complexity sample-generating oracle. This promise problem, in which we assume that the sample-generating oracle has low complexity, is a natural variant of the property testing framework defined in \cref{sec:preliminaries}. It is motivated by our result from the preceding section, which showed that for some purposes, it suffices to restrict attention to oracles defined by low-complexity target functions. Later, in \cref{sec:constant-sample-proofs}, we will show how to use $\widetilde{T}$ to establish the structured symmetry of a property under consideration.

The main result of this section, which we state formally in \cref{thm:tester-sim}, will require the notion of \emph{consistency indicators} of a distinguisher family.

\begin{definition}[Consistency Indicators]
\label[definition]{def:ci}
    Given a function family $\cF \subseteq \{\cX \to [0, 1]\}$, consider the family of \emph{consistency indicators} $\Gamma_m(\cF) \subseteq \{(\cX \times \{0, 1\})^m \to \{0, 1\}\}$, which take as input $m$ labeled pairs $(x_i, y_i)$ and checks whether they are all consistent with some function in $\cF$, after thresholding:
    \[
        \Gamma_m(\cF) = \Bigl\{(x, y) \mapsto \bm{1}\bigl[\forall i \in [m],\, y_i = \bm{1}[f(x_i) \ge t_i]\bigr]  \,\Big|\, f \in \cF,\,t_1, \ldots, t_m \in \R\Bigr\}.
    \]
    As usual, $(x, y)$ is our abbreviation for the $m$-tuple of pairs $(x_i, y_i)$.    
\end{definition}

Much like our proof of the oracle simulation lemma in \cref{sec:oracle-sim}, our proof of \cref{thm:tester-sim} is based on indistinguishability. While we will eventually take the simulator in \cref{sec:oracle-sim} to be the one provided by the complexity-theoretic regularity lemma, the simulator in the present section will eventually come from the supersimulators lemma. Interestingly, in \cref{sec:oracle-sim}, the oracle played the role of the object to be simulated, and the property tester played the role of the distinguisher. In contrast, in this section, the tester shall play the role of the object to be simulated, and the oracle shall play the role of the distinguisher. 

\begin{lemma}[Tester Simulation]
\label[lemma]{thm:tester-sim}
    Fix a distribution $\cD \in \Delta(\cX)$ and a family $\cF \subseteq \{\cX \to [0, 1]\}$. If $\widetilde{T}$ is a $(\Gamma_m(\cF), \gamma)$-regular simulator for $\bar{T} : (\cX \times \{0, 1\})^m \to [0, 1]$, then for all $\tilde{f} \in \cF$,
    \[
        \Bigl\lvert\E\bigl[\bar{T}(x, \tilde{y})\bigr] - \E\bigl[\widetilde{T}(x, \tilde{y})\bigr]\Bigr\rvert \le 2^m \gamma,
    \]
    where $x_1, \ldots, x_m \iid \cD$ and $\tilde{y}_i|x_i \sim \cB(\tilde{f}(x_i))$.
\end{lemma}

\Cref{thm:tester-sim} can be viewed as a generalization of a technical result from \cite{blais2019testable}. While our result applies to a general family $\cF \subseteq \{\cX \to [0, 1]\}$, the technical result from \cite{blais2019testable} corresponds to the special case of the class $\cF_0 = \{\cX \to \{0, 1\}\}$. Note that $\cF$ may contain $[0, 1]$-valued functions, which is why we need to incorporate the notion of consistency indicators $\Gamma_m$ after thresholding. These considerations were not needed for the case of $\cF_0$, which contains only Boolean functions.

While \cite{blais2019testable} obtain their regular approximation to $T$ using a hypergraph regularity lemma, we will eventually acquire $\widetilde{T}$ from the supersimulators lemma, thus circumventing the hypergraph-based formalism.

\begin{proof}[Proof of \Cref{thm:tester-sim}]
    We first rewrite the labels $\tilde{y}_i$ in a convenient form. Specifically, we write
    \[
        \tilde{y}_i = \bm{1}[\tilde{f}(x_i) \ge t_i]
    \]
    for a sequence of uniformly random thresholds $t_1, \ldots, t_m \sim [0, 1]$, which are independent of each other and $x_1, \ldots, x_m$. Since each $t_i$ is uniform, it is clear that $\tilde{y}_i|x_i \sim \cB(\tilde{f}(x_i))$, as required. Next, we relate $\bar{T}(x, \tilde{y})$ to $\bar{T}(x, z)$, where $z_1, \ldots, z_m \in \{0, 1\}$ are independent and uniformly random labels. To do so, let us condition on the values of $x$ and $t$. Then, there is a $2^{-m}$ probability that $z_i = \tilde{y}_i$ for all indices $i \in [m]$. Therefore, if we condition on the values of $x$ and $t$, we have
    \begin{equation}
    \label{eq:conditioning}
        \bar{T}(x, \tilde{y}) = 2^m \cdot \E_{z}\bigl[\bar{T}(x, z)\, \bm{1}[\tilde{y} = z] \,\big|\, x, t\bigr].
    \end{equation}
    Next, observe that the indicator function $\bm{1}[\tilde{y} = z]$ can be written as $\bm{1}[z = \bm{1}[\tilde{f}(x) \ge t]]$. (Here, the condition $\tilde{f}(x) \ge t$ should be read coordinate-wise.) When viewed as a function of $x$ and $z$, this function belongs to the distinguisher family $\Gamma_m(\cF)$, by definition. Note that this is a \emph{random function}, which depends on the thresholds $t$. By assumption, $\widetilde{T}$ is a $(\Gamma_m(\cF), \gamma)$-regular simulator for $\bar{T}$. Therefore, taking the expectation over $x$,
    \[
        \Bigl\lvert\E_{x,z}\bigl[\bar{T}(x, z) \bm{1}[\tilde{y} = z] \,\big|\, t\bigr] - \E_{x,z}\bigl[\widetilde{T}(x, z) \bm{1}[\tilde{y} = z] \,\big|\, t\bigr]\Bigr\rvert \le \gamma.
    \]
    Finally, taking the expectation over $t$ and using equation \eqref{eq:conditioning} yields
    \[
        \Bigl\lvert\E_{x,t}\bigl[\bar{T}(x, \tilde{y}) \bigr] - \E_{x,t}\bigl[\widetilde{T}(x, \tilde{y})\bigr]\Bigr\rvert \le 2^m\gamma. \qedhere
    \]
\end{proof}

\subsection{Constructing the Partition}
\label[section]{sec:constant-sample-proofs}

In this section, we prove \cref{thm:main-hard,thm:main-easy}. First, we instantiate \cref{thm:oracle-sim} (oracle simulation) with \cref{thm:ttv} (complexity-theoretic regularity). Next, we instantiate \cref{thm:tester-sim} (tester simulation) with \cref{thm:super} (supersimulators). Finally, we turn our attention to small circuit testers and use the fact that operations like restrictions and thresholding barely increase circuit size.

\mainhard*

\begin{proof}
    Let $T : (\cX \times \{0, 1\})^m \times \{0, 1\}^\ell \to \{0, 1\}$ be a valid tester for $\cP$. This means that $T$ outputs \textsc{Accept} with probability at least $2/3$ when $f \in \cP$ and outputs \textsc{Reject} with probability at least $2/3$ when $f$ is $\eps$-far from $\cP$. As usual, these probabilities are with respect to independent samples $x_i \sim \cD$ with labels $y_i = f(x_i)$, and a uniform $\ell$-bit random seed $r \in \{0, 1\}^\ell$. As usual, let $\bar{T}(x, y)$ denote the expected value of $T(x, y, r)$ over $r$. We proceed in several steps.
    
    \paragraph{Step 1: Constructing the Supersimulator}
    In order to approximate $\cP$ by a property $\cQ$ with structured symmetry, we will first construct a supersimulator $\widetilde{T}$ for $T$, and then extract the desired partition from the inner structure of $\widetilde{T}$. To this end, fix $\delta,\gamma > 0$. Let $\cG$ be the growth function that takes as input a function $T' : (\cX \times \{0, 1\})^m \to [0, 1]$ and outputs the distinguisher family
    \[
        \cG(T') = \Gamma_m\Bigl(\cS_{<(2/\delta^2),(\delta/2)}\bigl(\cR(T) \cup \cR(T')\bigr)\Bigr).
    \]
    By \cref{thm:super} (supersimulators), $\bar{T}$ has a $(\cG(\widetilde{T}), \gamma)$-regular supersimulator
    \[
        \widetilde{T} \in \cS_{<(2/\gamma^2),(\gamma/2)}(\cG).
    \]
    By \cref{thm:ttv} (complexity-theoretic regularity), $f$ has an $(\cR(T) \cup \cR(\widetilde{T}), \delta)$-regular simulator
    \[
        \tilde{f} \in \cS_{<(2/\delta^2),(\delta/2)}\bigl(\cR(T) \cup \cR(\widetilde{T})\bigr).
    \]
    As usual, we let $\tilde{y}_i$ denote the modeled labels generated from $\tilde{f}$, which means that $\tilde{y}_i|x_i \sim \cB(\tilde{f}(x_i))$. 

    \paragraph{Step 2: Applying the Two Key Lemmas}
    Next, we will use \cref{thm:oracle-sim} (oracle simulation) and \cref{thm:tester-sim} (tester simulation) to show that $T$ and $\widetilde{T}$ have similar probabilities of outputting \textsc{Accept} regardless of the labeling function $f : \cX \to \{0, 1\}$. To make the argument more concise, we will write $a \approx_\rho b$ if two scalars $a,b \in \R$ differ in absolute value by at most $\rho$. First, since $\tilde{f}$ is a $(\cR(T), \delta)$-regular simulator for $f$, \cref{thm:oracle-sim} implies
    \[
        \E\bigl[T(x, y, r)\bigr] \approx_{2m \delta} \E\bigl[T(x, \tilde{y}, r)\bigr] = \E\bigl[\bar{T}(x, \tilde{y})\bigr].
    \]
    Next, since $\widetilde{T}$ is a $(\cG(\widetilde{T}), \gamma)$-regular simulator for $\bar{T}$, \cref{thm:tester-sim} implies
    \[
        \E\bigl[\bar{T}(x, \tilde{y})\bigr] \approx_{2^m\gamma} \E\bigl[\widetilde{T}(x, \tilde{y})\bigr].
    \]
    Finally, since $\tilde{f}$ is a $(\cR(\widetilde{T}), \delta)$-regular simulator for $f$, applying \cref{thm:oracle-sim} again yields
    \[
        \E\bigl[\widetilde{T}(x, \tilde{y})\bigr] \approx_{2m\delta} \E\bigl[\widetilde{T}(x, y)\bigr].
    \]
    Combining these three steps, we deduce that 
    \[
        \E\bigl[T(x, y, r)\bigr] \approx_{4m\delta + 2^m\gamma} \E\bigl[\widetilde{T}(x, y)\bigr].
    \]
    In other words, the probability that $T$ outputs \textsc{Accept} differs from the expected output of $\widetilde{T}$ by at most $4m\delta + 2^m\gamma$, which is less than $1/6$ for appropriately chosen $\delta = \Theta(1/m)$ and $\gamma = \Theta(1/2^m)$.

    \paragraph{Step 3: Defining the New Property}
    Having related $\widetilde{T}$ to $T$, we define the property
    \[
        \cQ = \biggl\{f : \cX \to \{0, 1\} \;\bigg|\; \E\bigl[\widetilde{T}(x, y)\bigr] \ge \frac{1}{2} \biggr\},
    \]
    where, as usual, the expectation is over $m$ independent samples $x_i \sim \cD$, where $y_i = f(x_i)$. In other words, we say that a function $f$ has the property $\cQ$ if the expected output of the simulated tester $\widetilde{T}$ is at least $1/2$ on samples labeled by $f$.
    
    We claim that $\cP \subseteq \cQ \subseteq \cP_\eps$. Indeed, this follows immediately from the previously established relationship between $T$ and $\widetilde{T}$. In slightly more detail, suppose that $f \in \cP$. Then, since $T$ is a valid tester for $\cP$, it outputs \textsc{Accept} with probability at least $2/3$. Consequently, the expected output of $\widetilde{T}$ is at least $2/3 - 1/6 = 1/2$. Thus, $f \in \cQ$. A similar argument shows that $\cQ \subseteq \cP_\eps$. At this point, we have shown that $\cP \subseteq \cQ \subseteq \cP_\eps$, where membership in $\cQ$ depends only on the probability of acceptance by $\widetilde{T}$.
    
    \paragraph{Step 4: Defining the Partition} Finally, we show that $\cQ$ exhibits the desired structured symmetry. For this, recall that the supersimulator lemma states that $\widetilde{T}$ has the form
    \[
        \widetilde{T} = \widetilde{T}_k = \biggl[\frac{\gamma}{2}\bigl(F_1 + \cdots + F_{k}\bigr)\biggr]_0^1
    \]
    for some $k < 2/\gamma^2$, where each term $F_j$ belongs to $\pm \cG(\widetilde{T}_{j-1})$. In other words, $F_j$ is a (possibly negated) consistency indicator for some function
    \[
        f_j \in \cS_{<(2/\delta^2),(\delta/2)}\bigl(\cR(T) \cup \cR(\widetilde{T}_{j-1})\bigr).
    \]
    Consequently, the expected output of $\widetilde{T}$ depends only on the density of the labeling function on the following collection of sets, for various indices $i \in [m]$ and $j \in [k]$ and fixed thresholds $t_{ij} \in [0, 1]$:
    \[
        S_{ij} = \{x \in \cX \mid f_j(x) \ge t_{ij}\}.
    \]
    Taking all $2^{mk} \le 2^{2^{O(m)}}$ intersections of the sets $S_{ij}$ yields a partition witnessing the $2^{2^{O(m)}}$-part symmetry of $\cQ$.
    
    Finally, we address the partition complexity of $\cQ$. For this, let $\cD$ be the uniform distribution on $\cX = \{0, 1\}^n$ and suppose that $T$ is computable by a circuit of size $s$. In this case, the restrictions $\cR(T)$ obtained by hard-wiring various inputs to $T$ are also computable by circuits of size $s$. Therefore, the subsequent \cref{thm:circuit-counting} (with our choices of $\delta = \Omega(1/m)$, $\gamma = \Omega(1/2^m)$, and $k = O(1/\gamma^2)$) shows that $\cQ$ has partition complexity at most $2^{O(m)}s$, concluding the proof of \cref{thm:main-hard}.
\end{proof}

\begin{lemma}[Counting Circuit Gates]
\label[lemma]{thm:circuit-counting}
    Fix $\delta,\gamma>0$, functions $f_j: \cX \to [0, 1]$, thresholds $t_{ij} \in [0, 1]$, and signs $\sigma_j \in \{-1, +1\}$ for indices $i \in [m]$ and $j \in [k]$. Consider any Boolean-valued function $T : (\cX \times \{0, 1\})^m \times \{0, 1\}^\ell \to \{0, 1\}$ and let the function $\widetilde{T}_k : (\cX \times \{0, 1\})^m \to [0, 1]$ be
    \[
        \widetilde{T}_k(x, y) = \biggl[\frac{\gamma}{2}\sum_{j=1}^k \sigma_j \bm{1}\bigl[\forall i \in [m], \, y_i = \bm{1}[f_j(x_i) \ge t_{ij}]\bigr]\biggr]_0^1.
    \]
    Let $\widetilde{T}_0 = 0$. If $f_j \in \cS_{<(2/\delta^2),(\delta/2)}(\cR(T) \cup \cR(\widetilde{T}_{j-1}))$ for each index $j \in \N$, then there exists a circuit, which we call the \emph{classifier circuit for $\widetilde{T}_k$}, which has the following properties:
    \begin{itemize}
        \item The classifier receives as input $p = \mathrm{poly}(mk\log(1/\gamma)/\delta)$ Boolean values $r_j(x)$ for some restriction functions $r_1, \ldots, r_p \in \cR(T)$ evaluated at some particular point $x \in \cX$;
        \item The classifier uses $q = \mathrm{poly}(mk\log(1/\gamma)/\delta)$ Boolean circuit gates;
        \item The classifier outputs the $mk$ Boolean values $\bm{1}[f_j(x) \ge t_{ij}]$ for all $i \in [m]$ and $j \in [k]$.
    \end{itemize}
\end{lemma}

\begin{proof}
    We will prove the lemma by induction on $k$. Suppose we have a classifier circuit for $\widetilde{T}_{k-1}$ using $p$ functions $r_1, \ldots, r_p$ from $\cR(T)$ and $q$ circuit gates. We will show how to construct a classifier circuit for $\widetilde{T}_k$ while increasing $p$ and $q$ by at most $\mathrm{poly}(mk\log(1/\gamma)/\delta)$ each.
    
    By assumption, our circuit already computes the values $\bm{1}[f_j(x) \ge t_{ij}]$ for all indices $i \in [m]$ and $j \in [k - 1]$. Therefore, we just need to modify the circuit to also compute the values $\bm{1}[f_k(x) \ge t_{ik}]$ for each index $i \in [m]$. To this end, recall that
    \[
        f_k \in \cS_{<(2/\delta^2),(\delta/2)}\bigl(\cR(T) \cup \cR(\widetilde{T}_{k-1})\bigr).
    \]
    This means that the function $f_k$ is a structured sum of at most $2/\delta^2$ restrictions of either $T$ or $\widetilde{T}_{k-1}$. We will consider each of these restrictions separately, depending on whether they came from $T$ or $\widetilde{T}_{k-1}$. For the restrictions that came from $T$, add each one to the existing list $r_1, \ldots, r_p$. This increases $p$, the length the list, by at most $2/\delta^2$. Next, consider a restriction $r$ obtained by hard-wiring all inputs to $\widetilde{T}_{k-1}$ except for $x_i$, for some index $i \in [m]$. Then, $r$ has the form
    \[
        r(x) = \biggl[\frac{\gamma}{2} \sum_{j=1}^{k-1} \sigma_j \bm{1}\big[y_i = \bm{1}[f_j(x) \ge t_{ij}] \text{ and } y_{i'} = \bm{1}[f_j(x_{i'}) \ge t_{i'j}] \text{ for all }i' \neq i\bigr] \biggr]_0^1,
    \]
    for some fixed sequence $x_{\neq i}$ comprising values $x_{i'}$ for all $i' \neq i$ and some fixed labels $y \in \{0, 1\}^m$. Notice that the truth values of the conditions $y_{i'} = \bm{1}[f_j(x_{i'}) \ge t_{i'j}]$ that appear in the formula for $r$ are fixed functions of $x_{\neq i}$ and $y$. Therefore, they too may be hard-wired in advance. Consequently, the formula for $r$ can be substantially simplified to only depend on some fixed subset of indices $I \subseteq [k-1]$ in the summation:
    \[
        r(x) = \biggl[\frac{\gamma}{2}\sum_{j \in I} \sigma_j \bm{1}\bigl[ y_i = \bm{1}[f_j(x) \ge t_{ij}]\bigr]\biggr]_0^1.
    \]
    By the inductive hypothesis, our existing classifier circuit has already computed the Boolean values $\bm{1}[f_j(x) \ge t_{ij}]$ for all $j \in I$. Therefore, using the formula above, we can compute $r(x)$ using just $\mathrm{poly}(k\log(1/\gamma))$ additional circuit gates. Since $f_k$ is a sum of $2/\delta^2$ of these restrictions, along with some others from the list $r_1, \ldots, r_p$, we deduce that $f_k$ can be computed using just $\mathrm{poly}(mk\log(1/\gamma)/\delta)$ additional circuit gates. Of course, once we have $f_k(x)$, the Boolean values $\bm{1}[f_k(x) \ge t_{ik}]$ for each index $i \in [m]$ are similarly inexpensive to compute.
\end{proof}

We conclude this section with a short proof of \cref{thm:main-easy}, the converse to \cref{thm:main-hard}.

\maineasy*

\begin{proof}
    We are given a partition of $\cX$ into sets $S_1, \ldots, S_k$, each of which has complexity at most $s$, such that whether or not $f$ belongs to $\cQ$ can be determined from its $k$ \emph{densities}
    \[
        \mu_j(f) = \E\bigl[f(x)\bm{1}[x \in S_j]\bigr].
    \]
    As usual, the above expectation is over $x \sim \cD$. Let $\mu(f) \in [0, 1]^k$ denote the vector of $k$ densities of $f$, and let $\hat\mu_m(f)$ denote its empirical estimate given $m$ independent samples $(x_i, f(x_i))$. By Hoeffding's inequality and a union bound over the $k$ sets, with probability at least $2/3$, we have $\norm{\hat\mu_m(f) - \mu(f)}_1 \le k\delta$ as long as $m = O(\log(k)/\delta^2)$.
    
    Consider the function $T$ which takes as input a $\delta$-granular density profile $v \in \{0, \delta, \ldots, 1 -\delta, 1\}^k$, and outputs $1$ if some $f \in \cQ$ has a nearly matching profile, meaning that $\norm{\mu(f) - v}_1 \le 2k\delta$. For brevity, let $\pi_\delta(v)$ denote the $\delta$-granular coordinate-wise rounding of $v \in [0,1]^k$ to multiples of $\delta$.

    Fix any $f : \cX \to \{0, 1\}$. With probability at least $2/3$, we have
    \[
        \norm{\mu(f) - \pi_\delta(\hat\mu_m(f))}_1 \le \norm{\mu(f) - \hat\mu_m(f)}_1 + \norm{\hat\mu_m(f) - \pi_\delta(\hat\mu_m(f))}_1 \le 2k\delta.
    \]
    Therefore, $f \in \cQ$ implies $T(\pi_\delta(\hat\mu_m(f)) = 1$. Conversely, if $T(\pi_\delta(\hat\mu_m(f))) = 1$, then for some $\tilde{f} \in \cQ$,
    \[
        \norm{\tilde{f} - f}_1 = \norm{\mu(\tilde{f}) - \mu(f)}_1 \le \norm{\mu(\tilde{f}) - \pi_\delta(\hat\mu_m(f))}_1 + \norm{\pi_\delta(\hat\mu_m(f)) - \mu(f)}_1 \le 4k\delta,
    \]
    so $f \in \cQ_{4k\delta}$. Setting $\delta = \eps/4k$, we see that running $T$ on the rounded empirical estimate $\pi_\delta(\hat\mu_m(f))$ yields a valid tester for $\cQ$ with proximity $\eps$. Since membership in each set $S_i$ can be computed with a circuit of size at most $s$, we can determine which parts contain each of the $m = (k/\eps)^{O(1)}$ samples with a circuit of size $(k/\eps)^{O(1)}s$. Maintaining empirical averages can similarly be done with a circuit of size $(k/\eps)^{O(1)}$. Finally, any post-processing function $T : \{0, \delta, \ldots, 1-\delta, 1\}^k \to \{0, 1\}$ with $\delta = \Omega(\eps/k)$ can be computed with at most $(k/\eps)^{O(k)}$ additional circuit gates by brute force.
\end{proof}

\subsection{Consistency Counting}
\label[section]{sec:cc}

In this section, we present our modification to \cref{thm:blais} that circumvents the $2^{2^{O(m)}}$ dependence on the tester's sample size $m$, which allows us to make meaningful statements about sample sizes larger than $O(\log n)$. Recall that our goal is to show that any $m$-sample testable property essentially boils down to what we call \emph{consistency counting} on a collection of $2^{O(m)}$ functions. We first recall the definition of consistency counting and the statement of our theorem.

\defcc*

\thmcc*

The proof of \cref{thm:cc} is broadly similar to that of \cref{thm:main-hard}, but it makes additional use of the additive structure of the simulator from the complexity-theoretic regularity lemma. Since we do not consider computational complexity in this section, we will not require the supersimulators lemma like we did in the previous section. Instead, our proof will only require the basic complexity-theoretic regularity lemma. Dissecting the proof further shows that we can always ensure $\cF_+ \subseteq \cP$ and $\cF_- \cap \cP_\eps = \varnothing$. 

\begin{proof}[Proof of \Cref{thm:cc}]
    As in the proof of \cref{thm:main-hard}, we are given a property $\cP \subseteq \{\cX \to \{0, 1\}\}$ and a tester $T_0 : (\cX \times \{0, 1\})^{m_0} \times \{0, 1\}^{\ell_0} \to \{0, 1\}$ that uses $m_0$ samples. Since $T_0$ is a valid tester for $\cP$, it has success probability at least $2/3$. By running $O(1)$ independent copies of $T_0$ and taking a majority vote, we can boost the success probability from $2/3$ to $11/12$ (or, for that matter, any constant that is strictly less than $1$). Call the resulting tester $T : (\cX \times \{0, 1\})^m \times \{0, 1\}^\ell \to \{0, 1\}$, where $m = O(m_0)$ and $\ell = O(\ell_0)$. 
    
    Next, define $x$, $y$, $r$, and $\bar{T}$ as in the proof of \cref{thm:main-hard}. In order to construct a consistency counter that is a valid tester for $\cP$, consider the distinguisher family $\cC \subseteq \{(\cX \times \{0, 1\})^m \to \{0, 1\}\}$ comprising all possible $m$-fold \emph{consistency} functions:
    \[
        \cC = \Bigl\{(x, y) \mapsto \bm{1}\bigl[y_1 = f(x_1) \land \cdots \land y_m = f(x_m)\bigr] \;\Big|\; f : \cX \to \{0, 1\}\Bigr\}.
    \]
    Next, fix $\gamma > 0$. By \cref{thm:ttv} (complexity-theoretic regularity), there exists a function $\widetilde{T} \in \cS_{k,\gamma/2}(\cC,1/2)$ that is a $(\cC, \gamma)$-regular simulator for $\bar{T}$, where $k < 2/\gamma^2$ and $1/2$ refers to the constant function that always outputs $1/2$. Then, by \cref{thm:tester-sim} (tester simulation), we have
    \[\E\bigl[T(x, y, r)\bigr] = \E\bigl[\bar{T}(x, y)\bigr] \approx_{2^m \gamma} \E\bigl[\widetilde{T}(x, y)\bigr]. \]
    In other words, replacing $T$ with $\widetilde{T}$ changes the expected output by at most $2^m\gamma$, which is strictly less than $1/12$ for appropriately chosen $\gamma = \Theta(1/2^m)$. Next, by the complexity-theoretic regularity lemma, we know that $\widetilde{T}$ has the form
    \[
        \widetilde{T} = \biggl[\frac{1}{2} + \frac{\gamma}{2}\bigl(\sigma_1F_1 + \cdots + \sigma_kF_{k}\bigr)\biggr]_0^1
    \]
    where each term $F_j$ belongs to the family $\cC$ and $\sigma_1, \ldots, \sigma_k \in \{\pm 1\}$ are arbitrary signs. By the definition of $\cC$, each function $F_j$ is testing for consistency with some function $f_j : \cX \to \{0, 1\}$.
    
    Define the set of ``good'' functions $\cF_+$ to be the functions $f_j$ for which $\sigma_j = +1$, and define the set of ``bad'' functions $\cF_-$ to be the set of functions $f_j$ for which $\sigma_j = -1$. Then, the output of $\widetilde{T}$ depends only on the difference between the number of good and bad functions that fit the observed sample:
    \[
        \widetilde{T}(x, y) = \biggl[\frac{1}{2} + \frac{\gamma}{2}\Bigl\lvert\big\{f \in \cF_+ \;\big|\; \forall i \in [m],\, y_i = f(x_i)\bigr \}\Bigr\rvert - \frac{\gamma}{2}\Bigl\lvert\big\{f \in \cF_- \;\big|\; \forall i \in [m],\, y_i = f(x_i)\bigr \}\Bigr\rvert\biggr]_0^1.
    \]
    In particular, $\widetilde{T}(x, y) > 1/2$ if and only if the $(m, \cF_+, \cF_-)$-consistency counter outputs \textsc{Accept}.
    
    Finally, we show that this consistency counter is a valid tester for $\cP$. For this, suppose that $f \in \cP$. By assumption, $T$ has success probability at least $11/12$, and the expected output of $\widetilde{T}$ differs from that of $T$ by at most $1/12$. Therefore, the probability that the consistency counter mistakenly outputs \textsc{Reject} is
    \[
        \Pr\Bigl[\widetilde{T}(x, y) \le \frac{1}{2} \Bigr] = \Pr\Bigl[1 - \widetilde{T}(x, y) \ge \frac{1}{2} \Bigr] \le 2 \cdot \E\bigl[1 - \widetilde{T}(x, y) \bigr] \le 2 \cdot \Bigl(\E\bigl[1 - T(x, y, r)] + \frac{1}{12}\Bigr) \le \frac{1}{3}.
    \]
    Similarly, if $f \notin \cP_\eps$, then the consistency counter mistakenly outputs \textsc{Accept} with probability at most $1/3$. We conclude that the $(m, \cF_+, \cF_-)$-consistency counter is a valid tester for $\cP$.
\end{proof}

\subsection{Dense Distribution Testing}
\label[section]{sec:dt}

In this section, we discuss our extension of \cref{thm:main-hard} from sample-testing Boolean functions to testing properties of dense distributions. First, we review the relevant notions of density and entropy. Given a finite domain $\cX$, fix a distribution $\cD_0 \in \Delta(\cX)$, which we call the \emph{base distribution} or \emph{reference distribution}. For example, when $\cX = \{0, 1\}^n$, we take $\cD_0$ to be the uniform distribution on $\cX$. Given another distribution $\cD$ whose support is contained in that of $\cD_0$, we say that $\cD$ is \emph{$\mu$-dense} in $\cD_0$ if the ratio of $\cD$ to $\cD_0$ is at most $1/\mu$ pointwise. For example, when $\cX = \{0, 1\}^n$, a distribution is $\mu$-dense if it is a uniform distribution over a set of size $\mu\abs{\cX}$, or a convex combination thereof. In this case, we equivalently say that the distribution has \emph{min-entropy} at least $k = n - \log(1/\mu)$.

For another example, suppose that we sample $x \sim \{0, 1\}^n$ uniformly, and then set $y = f(x)$ for some fixed Boolean function $f$. Then, the joint distribution of the pair $(x, y)$ is $\mu$-dense in the uniform distribution over $\{0, 1\}^{n+1}$ for $\mu = 1/2$. Indeed, this is precisely the fact that we used in the proof of \cref{thm:tester-sim}, which lead to the $(1/\mu)^m = 2^m$ blowup in the regularity parameter $\gamma$. In fact, a lower bound on the density of the distribution of $(x, y)$ was essentially the only thing our proof needed; it was not important for $x$ to be uniform and $y$ to be a deterministic function of $x$.

With the aforementioned example in mind, we now define a natural analogue of the property testing framework from \cref{sec:preliminaries} that applies not just to uniform $n$-bit strings with Boolean labels, but also more general dense distribution properties. For this, recall that the total variation distance between two distributions $\cD_0,\cD_1 \in \Delta(\cX)$ is their maximum absolute difference over any subset of the domain: $\max_{S \subseteq \cX}\, \abs{\cD_0(S) - \cD_1(S)}$.

\begin{definition}[Densely Testable Property]
    Let $\cP \subseteq \Delta(\{0, 1\}^n)$ be a distribution property. We say that $\cP$ is \emph{$\mu$-densely testable} with proximity parameter $\eps > 0$ if there exists a randomized circuit $T$ of size $s$ that receives as input $m$ samples $x_i \sim \cD$ for some $\cD \in \Delta(\{0, 1\}^n)$, always outputs either \textsc{Accept} or \textsc{Reject}, and meets the following two requirements:
    \begin{itemize}
        \item If $\cD \in \cP$ and $\cD$ is $\mu$-dense, then $T$ outputs \textsc{Accept} with probability at least $2/3$.
        \item If $\cD$ is $\eps$-far from $\cP$ and $\cD$ is $\mu$-dense, then $T$ outputs \textsc{Reject} with probability at least $2/3$.
    \end{itemize}
    In each condition, the probability is computed over randomness in the sample and internal to $T$.
\end{definition}

In our to state our generalized main result for dense distribution testing, we will again need to invoke the notion of ``structured symmetry'' that we have been using throughout this paper. However, we emphasize that in the context of distribution testing, this is \emph{not} the same as the more standard notion of label-invariance. Rather, we say that a distribution property $\cP$ is \emph{$k$-part symmetric} if there is a partition of $\{0, 1\}^n$ into disjoint parts $S_1, \ldots, S_k$ such that $\cP$ is invariant under any redistribution of probability mass within each part. (More formally, whether or not $\cD \in \cP$ should depend only on the $k$ densities $\cD(S_i)$ for $i \in [k]$.)

We now recall the statement of our result on $\mu$-dense distribution testing. The converse, of course, is straightforward. Since the proof of \cref{thm:distr-hard} is extremely similar to that of \cref{thm:main-hard}, which we carried out in detail in \cref{sec:symmetry-thm}, we simply sketch the main alterations to the argument.

\distrhard*

\begin{proof}[Proof Sketch]
    First, some notation. Given $f : \{0, 1\}^n \to [0, 1/\mu]$, let $\cD_f$ denote the $\mu$-dense distribution over $\{0, 1\}^n$ with mass function $f(x)2^{-n}$. Given $x \in (\{0, 1\}^{n})^m$, let $f^{(m)}(x) = \prod_{i=1}^m f(x_i)$.
    
    Our strategy will be to generalize \cref{thm:oracle-sim} (oracle simulation) and \cref{thm:tester-sim} (tester simulation), which were the key building blocks of \cref{thm:main-hard}. To generalize the oracle simulation lemma, consider a tester $T : (\{0, 1\}^n)^m \times \{0, 1\}^\ell\to \{0, 1\}$ computable by a circuit of size $s$. If $\mu\tilde{f}$ is $(s, \delta)$-indistinguishable from $\mu f$, then by the same hybrid argument as before,
    \[
       \mathbb{E}_f\bigl[T(x,r)\bigr] = \frac{1}{\mu}\E\bigl[\mu f^{(m)}(x) T(x,r)] \approx_{\frac{m\delta}{\mu}} \frac{1}{\mu}\E\bigl[\mu \tilde{f}^{(m)}(x)T(x,r)\bigr] = \mathbb{E}_{\tilde{f}}\bigl[T(x,r)\bigr],
    \]
    where $\mathbb{E}_f$ denotes the expectation over $x \sim \cD_f^m$ and $\E$ denotes the expectation over uniform $x$. Next, to generalize the tester simulation lemma, suppose that $\widetilde{T}$ and $\bar{T}$ (the mean function of $T$) are $\gamma$-indistinguishable with respect to $m$-fold thresholds of $\tilde{f}$. Then, for a uniform $t \sim [0, 1]^m$,
    \[
        \mathbb{E}_{\tilde{f}}\bigl[\bar{T}(x)\bigr] = \mu^{-m}\E\bigl[\bm{1}[\mu \tilde{f}(x) \ge t]\,\bar{T}(x)\bigr] \approx_{\mu^{-m}\gamma} \mu^{-m}\E\bigl[\bm{1}[\mu \tilde{f}(x) \ge t]\, \widetilde{T}(x)\bigr] = \mathbb{E}_{\tilde{f}}\bigl[\widetilde{T}(x)\bigr].
    \]
    
    Having generalized the two lemmas, we proceed as in the proof of \cref{thm:main-hard}. Set $\delta = \Theta(\mu/m)$ and $\gamma = \Theta(\mu^m)$. By \cref{thm:super} (supersimulators), there exists a circuit $\widetilde{T}$ of size $\tilde{s}$ that fools all $m$-fold thresholds of functions $\tilde{f}$ computable by a circuit of size at most $O(\max(s, \tilde{s})/\delta^2)$. Next, by \cref{thm:ttv} (complexity-theoretic regularity), given any $\mu$-dense distribution $\cD_f$, the function $\mu f$ is $(\max(s, \tilde{s}), \delta)$-indistinguishable from some $\mu\tilde{f}$ computable in size $O(\max(s,\tilde{s})/\delta^2)$. Thus, combining the two generalized lemmas, we have that, say, $\abs{\mathbb{E}_f[T(x,r)] - \mathbb{E}_f[\widetilde{T}(x)]} < 1/6$. This implies that $\cP \subseteq \cQ \subseteq \cP_\eps$, where $\cQ$ is the property that the expected output of $\widetilde{T}$ is at least $1/2$. As before, the structured symmetry can be read off from the structure of the circuit $\widetilde{T}$ provided by the supersimulators lemma. The bounds $2^{(1/\mu)^{O(m)}}$ and $(1/\mu)^{O(m)}s$ on the number of parts and the partition complexity, respectively, follow from the same counting arguments as before with our new choices of the parameters $\delta = \Theta(\mu/m)$ and $\gamma = \Theta(\mu^m)$.
\end{proof}

\section{Node-Level Private Graph Testing}
\label[section]{sec:dp-graphs}

In this section, we prove our main results regarding node-level private graph property testing, including our efficient node-level DP algorithm for computing graph regularity templates, as well as our private and function testing versions of the \cite{alon2006regularity} characterization of testability. We begin with \cref{thm:dp-partition}, which comprises the privacy and utility analyses of our sublinear time algorithm for privately computing a regularity template that a given input graph nearly satisfies (\cref{alg:dp-partition}).

\dppartition*

\begin{proof}
    One of the parameters of \cref{alg:dp-partition}, namely $\gamma$, is given in the statement of \cref{thm:dp-partition}. We shall define the algorithm's remaining parameters $k, q, \lambda$ (as well as two auxiliary parameters $r,\delta$ needed for this proof) as functions of $\gamma$ and $\eps$. First, we let $\delta = \lambda = (\eps\gamma)^{100}$. Next, we define $k = O_\gamma(1)$ as in the Szemer\'{e}di regularity lemma (\cref{thm:sz}). Next, we let $\abs{\cR}$ denote the size of the set $\cR$ constructed by \cref{alg:dp-partition} when run with parameters $\gamma, k, \lambda$. Finally, we let $r = \abs{\cR}/(\eps\gamma)^{100}$ and define $q = O_{r,\delta}(1)$ as in \cref{thm:reg-samp} (in particular, $q \ge r$). With these parameters in mind, we proceed to analyze the privacy and utility of the algorithm separately.
    
    \paragraph{Privacy Analysis}
    
    Consider any two node-adjacent graphs $G$ and $G'$ on $n$ vertices, and suppose that the differing vertex appears in the set $Q$ of $q$ subsampled vertices. Then, the induced subgraphs $G[Q]$ and $G'[Q]$ may only differ on the edges adjacent to this particular vertex. It follows that the loss function $\ell(R)$ used in \cref{alg:dp-partition} to instantiate the exponential mechanism, which counts the number of edge edits required to satisfy $R$, has \emph{sensitivity} at most $q$. This means that replacing $G[Q]$ with $G'[Q]$ can change $\ell(R)$ by at most $q$. We deduce from the privacy guarantee of the exponential mechanism (\cref{thm:expmech}) that our algorithm satisfies node-level $\alpha$-differential privacy for $\alpha = O(1)$. Finally, since $Q$ contains only a random sample of $q$ out of the $n$ vertices of the graph, privacy amplification by subsampling (\cref{thm:amp}) further reduces this privacy loss parameter to $\alpha = O(q/n) = O_{\gamma,\eps}(1/n)$, as desired.

    \paragraph{Utility Analysis}
    
    Recall that \cref{alg:dp-partition} samples a subset $Q$ of $q$ vertices from $G$ and operates on the subgraph $G[Q]$ that they induce. By \cref{thm:sz}, $G[Q]$ has some size-$k$ $\gamma$-regular equipartition $\cU = \{U_1, \ldots, U_k\}$. We will use $\cU$ to define a template $R = (\gamma, k, \eta, \bar{R})$ that is close to being satisfied by $G[Q]$ and that belongs to the finite set $\cR$ which the algorithm constructs and uses to choose its output. For this, we let $\eta_{ij}$ be equal to the density $d(U_i, U_j)$ after rounding to the nearest multiple of $\lambda$. Also, we let $\bar{R}$ be the set of indices of pairs of parts of $\cU$ that fail to be $\gamma$-regular. Clearly, $R \in \cR$, and by \cref{thm:tweak}, $G[Q]$ is $O(\lambda/\gamma^2)$-close to satisfying $R$. By the utility guarantee of the exponential mechanism (\cref{thm:expmech}), up to some additive slack, $G[Q]$ will be no further from satisfying whatever template $R' = (\gamma, k, \eta', \bar{R}')$ is actually output by the algorithm. More precisely, with probability at least $9/10$ over the random sample of $R'$ from $\cR$, the subgraph $G[Q]$ is $\tau$-close to some graph $H$ satisfying $R'$, where $\tau = O(\lambda/\gamma^2 + (\log\abs{\cR})/ q)$. For our choice of the parameters $\lambda$ and $q$, the slack $\tau$ is at most, say, $O((\eps\gamma)^{50})$. This is sufficiently small to carry out the remainder of the proof in a similar manner to the non-private case in \cite{alon2006regularity}. We include the argument for the sake of completeness:

    Next, we argue that any equipartition $\cU' = \{U'_1, \ldots, U'_k\}$ of the vertices of $H$ witnessing that $H$ satisfies $R'$ must also be a $\gamma'$-regular equipartition of $G[Q]$, where $\gamma' = \gamma + O(\sqrt{\tau}/\gamma^2)$. For this, recall that $G[Q]$ and $H$ are $\tau$-close, which means that we can transform $H$ back into $G[Q]$ by performing at most $\tau q^2$ edge additions or deletions. Consider how these edits might be distributed across the pairs of parts of $\cU'$. Clearly, there can be no more than $\sqrt{\tau} k^2$ pairs of parts that receive over $\sqrt{\tau} (q/k)^2$ edits, or else there would be over $\tau q^2$ edits in total. We refer to the pairs that receive more than this number of edits as \emph{damaged} pairs and denote their indices by $D \subseteq [k]^2$. For any undamaged pair $(i, j) \notin D$, where a $\sqrt{\tau}$ or smaller fraction of edge edits occurred, we have $\abs{d_{G[Q]}(U'_i, U'_j) - \eta_{ij}} \le \sqrt{\tau}$. Similarly, if an undamaged pair $(U_i', U_j')$ is $\gamma$-regular in $H$, then the transformation from $H$ back to $G[Q]$ worsens the regularity parameter of the pair to at most $\gamma + O(\sqrt{\tau}/\gamma^2)$. Indeed, this follows from the definition of regularity (\cref{def:graph-regularity}), because performing up to $\sqrt{\tau}(q/k)^2$ edge edits between any $A \subseteq U_i'$ and $B \subseteq U_j'$ with sizes at least $\gamma(q/k)$ changes the density of edges from $A$ to $B$ by at most $\sqrt{\tau}/\gamma^2$. Next, since at most a $\gamma$ fraction of pairs of $\cU'$ fail to be $\gamma$-regular in $H$, and since at most a $\sqrt{\tau}$ fraction of these pairs are damaged, we deduce that at most a $\gamma + \sqrt{\tau} \le \gamma'$ fraction of the pairs fail to be $\gamma'$-regular in $G[Q]$.
    
    By \cref{thm:reg-samp}, with probability at least $2/3$ over the random sample $Q$, there exists an equipartition $\cV = \{V_1, \ldots, V_k\}$ of $G$ that is \emph{$\delta$-similar} to the equipartition $\cU'$ of $G[Q]$ (see \cref{def:reg-sim} for the definition of $\delta$-similarity). We will use $\cV$ to show that $G$ is close to satisfying the template $R'$ output by the algorithm. Indeed, $\delta$-similarity implies that $\abs{d_G(V_i, V_j) - d_{G[Q]}(U'_i, U'_j)} \le \delta$ for all pairs $(i, j) \in [k]^2$. Since $d_{G[Q]}(U_i', U_j')$ is within $\sqrt{\tau}$ of $\eta'_{ij} = d_H(U_i', U_j')$ for undamaged pairs $(i, j) \notin D$, it follows from the triangle inequality that $\abs{d_G(V_i, V_j) - \eta'_{ij}} \le \delta + \sqrt{\tau}$ for undamaged pairs. We also know from $\delta$-similarity that $(V_i, V_j)$ is $(\gamma' + \delta)$-regular whenever $(U'_i, U'_j)$ is $\gamma'$-regular. Thus, we can transform $G$ into a graph satisfying $R'$ by first arbitrarily repairing each damaged pair to satisfy $\gamma$-regularity and then applying \cref{thm:tweak} yet again to repair the remaining pairs. The first step requires at most $\sqrt{\tau} n^2$ edge edits because of the limited number of damaged pairs, and the second step requires at most $O(\delta/\gamma^2 + \sqrt{\tau}/\gamma^4) \cdot n^2$ edits by \cref{thm:tweak}. Substituting our choices of $\delta$ and $\tau$, which in turn depend on our choices of $q,\lambda$, we see that the total number of edge edits required to make $G$ satisfy $R'$ is at most $\eps n^2$, as desired.
\end{proof}

Next, we prove \cref{thm:dp-graph-testing}, which extends the main characterization of testable graph properties from \cite{alon2006regularity} (i.e. \cref{thm:non-dp-graph-testing}) to the differentially private case. There are at least two natural approaches to proving a private version of the characterization. The first is to modify the proof from \cite{alon2006regularity} to add noise in various places. Specifically, one could replace the non-private tester for satisfying a fixed template $R$ (a key step in the proof of \cref{thm:non-dp-graph-testing}) with a node-level differentially private tester for satisfying $R$ (a slight modification of our \cref{alg:dp-partition}). The second approach is to directly reduce to \cref{thm:non-dp-graph-testing}. We take the latter approach for the sake of simplicity.

\dptesting*

\begin{proof}
    Consider a non-private tester in the form given by \cref{thm:graph-canonical}, which operates on a uniformly random subset of $q$ vertices from $G$. Such a tester can be made to satisfy node-level $O(1)$-differential privacy by negating the tester's output with some small constant probability (decreasing the tester's accuracy from $2/3$ to, say, $3/5$). Privacy amplification by subsampling (\cref{thm:amp}) further reduces the privacy loss parameter from $O(1)$ to $O(q/n)$, which is $O_\eps(1)$ whenever $q = O_\eps(1)$. At this point, we have shown that a property of $n$-vertex graphs is testable (with proximity $\eps$ using $O_\eps(1)$ queries, for all $\eps > 0$) if and only if it is privately testable (with privacy loss $O_\eps(1/n)$). We conclude by \cref{thm:non-dp-graph-testing}, which equates (non-private) testability to regular-reducibility.
\end{proof}

\subsection{Regularity Templates for Functions}
\label[section]{sec:rt}

We now turn our attention to proving a function testing analogue of the characterization of testability in \cite{alon2006regularity}. We begin by recalling our definition of compatibility with a regularity template, as well as the statement of the forward direction of our partial equivalence.

\compatibility*

\templateshard*

\begin{proof}
    We are given a property $\cP$ and a valid $m$-sample, size-$s$ tester $T$. Set $\delta = 1/13m$. By \cref{thm:ttv} (complexity-theoretic regularity), every Boolean function has an $(s, \delta)$-regular simulator computable by a circuit of size $O(m^2 s)$. With this in mind, let $\cT$ be the set of $(s, \delta)$-regular simulators of size $O(m^2 s)$ for functions in $\cP$. Clearly, $\cP \subseteq \cT_{s,\delta}$.

    Conversely, we claim that $\cT_{s, \delta} \subseteq \cP_\eps$. To prove this, suppose for the sake of contradiction that there exists a function $f \in \cT_{s,\delta} \setminus \cP_\eps$. Let $\tilde{f}$ be a function in $\cT$ that is $(s,\delta)$-indistinguishable from $f$, and let $f'$ be a function in $\cP$ that is $(s, \delta)$-indistinguishable from $\tilde{f}$. Since $T$ has a circuit of size $s$, so does every function in the class $\cR(T)$ of restrictions of $T$. Thus, $(s, \delta)$-indistinguishability implies $(\cR(T), \delta)$-indistinguishability. Therefore, letting $\bar{T}$ denote the expectation of $T$ over its internal randomness as usual, \cref{thm:oracle-sim} (oracle simulation) implies
    \[
        \E\bigl[T(x,f(x),r)\bigr] \approx_{2m\delta} \E\bigl[\bar{T}(x,\tilde{y})\bigr] \approx_{2m\delta} \E\bigl[T(x,f'(x),r)\bigr] \ge \frac{2}{3},
    \]
    where the expectation is over independent $x_i \sim \{0, 1\}^n$ and $\tilde{y}_i|x_i \sim \cB(\tilde{f}(x))$ and random seed $r$. In other words, replacing $f$ with $f'$ changes the probability that $T$ outputs \textsc{Accept} by at most $4m\delta$, which is strictly less than $1/3$ for our choice of $\delta$. However, since $f \notin \cP_\eps$, we also have $\E[T(x,f(x),r)] \le 1/3$, which is a contradiction. We conclude that $\cP \subseteq \cT_{s,\delta} \subseteq \cP_\eps$.
\end{proof}

We now prove our partial converse. The main idea behind the proof is a trick involving convex combinations, which was used in \cite{alon2006regularity} to prove that compatibility with a fixed regularity template is testable. The rest is a series of standard applications of concentration inequalities.

\templateseasy*

\begin{proof}
    Let $g$ denote the function we wish to test. For each $h \in \cT$, size-$s$ $f$, and $\sigma \in \{\pm 1\}$, consider
    \[
        \E\Bigl[\sigma f(x)\bigl(g(x)-h(x)\bigr)\Bigr].
    \]
    By definition, $g$ has the property $\cT_{s,\delta}$ if and only if there exists some template $h \in \cT$ such that the above quantity is at most $\delta$ for all size-$s$ $f$ and $\sigma \in \{\pm 1\}$. For now, consider a fixed $h \in \cT$. Note that there are at most $2^{O(s \log s)}$ circuits $f$ of size at most $s$. Therefore, by Hoeffding's inequality, we can estimate the quantity displayed above for all size-$s$ functions $f$ and signs $\sigma$ up to error $\alpha$ with failure probability $\beta$, given $m = O((s \log(s) + \log(1/\beta))/\alpha^2)$ labeled samples. Let us output \textsc{Accept} if all of these estimates are at most $\delta + \alpha$, and output \textsc{Reject} otherwise. Clearly, if $g \in \cT_{s,\delta}$, then the tester accepts with probability at least $1 - \beta$.
    
    Conversely, we claim that with probability at least $1 - \beta$, if the tester accepts, then $g$ is $\eps$-close to $\cT_{s,\delta}$. To prove this, first note that $g$ is $(s, \delta + \alpha)$-indistinguishable from $h$. Consequently, any convex combination $\lambda h + (1 - \lambda)g$ is $(s, (1-\lambda)(\delta + \alpha))$-indistinguishable from $h$. By performing randomized rounding on the output of this combination, we obtain a Boolean-valued function $\tilde{g}$ that is $2\lambda$-close to $g$ with probability at least $1 - \beta$, for a sufficiently large domain size $\abs{\cX} = O(\log(1/\beta))$. Moreover, this $\tilde{g}$ is $(s, (1-\lambda)(\delta+\alpha)+\gamma)$-indistinguishable from $h$ with probability at least $1 - \beta$, for $\abs{\cX} = O((s \log(s) + \log(1/\beta))/\gamma^2)$. If we set $\lambda = (\alpha + \gamma)/(\delta + \alpha)$ and $\gamma \le \alpha \le \eps\delta/4$, then $\lambda \le \eps/2$, in which case $\tilde{g}$ is both $\eps$-close to $g$ and $(s, \delta)$-indistinguishable from $h$, as desired.

    Until now, we have considered a fixed template $h \in \cT$ and used only $O((s \log(s)+\log(1/\beta))/\eps^2\delta^2)$ samples. However, since every circuit in $\cT$ has size at most $t$, we can run our tester for all $h \in \cT$, at the cost of a $2^{O(t \log t)}$ factor blowup in the failure probability $\beta$. Since $s \le t$, we conclude that $\cT_{s,\delta}$ is testable with proximity $\eps$ from $O(t\log(t)/\eps^2\delta^2)$ samples.
\end{proof}

\section{Simple Hypothesis Testing}
\label[section]{sec:mpv}

In this section, we prove our two results on the computational indistinguishability of product distributions. In \cref{sec:mpv-cma}, we prove our result using calibrated multiaccuracy, \cref{thm:mpv-cma-succinct}, which tightens the complexity gap in the multicalibration-based main result of \cite{marcussen2024characterizing}. In \cref{sec:mpv-super}, we prove our result based on supersimulators, \cref{thm:mpv-super-succinct}, which closes the gap entirely. We remark that all results in this section, although stated in terms of the $k$-fold powers of a single distribution, can be easily generalized to product distributions with $k$ distinct factors.

\subsection{Characterization via Calibrated Multiaccuracy}
\label[section]{sec:mpv-cma}

In this section, we prove \cref{thm:mpv-cma-succinct}, which quantitatively improves the main result of \cite{marcussen2024characterizing}. Essentially, our improvement comes from replacing the use of multicalibration with calibrated multiaccuracy. For background on (multi)calibration and multiaccuracy, see \cref{sec:preliminaries}. For convenience, we recall our theorem statement here.

\mpvcma*

The theorem has two parts, which we will prove separately. The first part involves two proxy distribution $\tcD_0$ and $\tcD_1$, and the second part has just one proxy distribution $\tcD_1$ (since it fixes $\tcD_0 = \cD_0$). Although both parts are stated in terms of size-bounded Boolean circuits, in what follows, we will prove both parts in the setting of an arbitrary distinguisher family. For the first part, our proof is a tighter and simpler analysis of essentially the same construction as in \cite{marcussen2024characterizing}, suggesting that their analysis may have only superficially required the stronger assumption of multicalibration. However, for the second part, we must introduce a new construction that was not present in \cite{marcussen2024characterizing}. Indeed, the corresponding construction and analysis in \cite{marcussen2024characterizing} genuinely requires the full strength of multicalibration.

\subsubsection{First Part of \cref{thm:mpv-cma-succinct}}

In what follows, we say that $\cD_0, \cD_1 : \cX \to \R$ are $(\cF, \eps)$-indistinguishable if for all $f \in \cF$,
\[
    \biggl\lvert\sum_{z \in \cX} f(z) \bigl(\cD_0(z) - \cD_1(z)\bigr)\biggr\rvert \le \eps.
\]
Given two distributions $\cD_0$ and $\cD_1$, our goal is to find an $(\cF, \eps)$-indistinguishable proxy for each, which we call $\tcD_0$ and $\tcD_1$, respectively, such that the information-theoretic distinguishability between ($k$-fold products of) the proxy pair roughly matches 
the computational distinguishability between ($k$-fold products of) the original pair.

To build intuition for the construction of $\tcD_0$ and $\tcD_1$, consider the following game, which is standard in cryptography. The objective of the game is to guess the value of an unknown, uniformly random bit $y \in \{0, 1\}$ given only a single observation $x \in \cX$, which is sampled from the distribution $\cD_0$ if $y = 0$ or from the distribution $\cD_1$ if $y = 1$. Equivalently, the pair $(x, y)$ is generated by first sampling $x$ from the balanced mixture $(\cD_0 + \cD_1)/2$ and then sampling $y$ from its conditional distribution given $x$, namely $\cB(g(x))$ for $g = \cD_1/(\cD_0 + \cD_1)$.

Keeping the marginal distribution of $x$ fixed and replacing the function $g$ with some efficient simulator $h$ yields a pair $(x, \tilde{y})$ with a possibly different distribution. Nevertheless, we can again view the generation of $(x, \tilde{y})$ in two equivalent ways, depending on which of $x$ or $\tilde{y}$ we sample first. Returning to the view in which the bit $\tilde{y} \in \{0, 1\}$ is sampled first, we \emph{define} the proxies $\tcD_0$ and $\tcD_1$ to be the conditional distribution of $x$ given $\tilde{y} = 0$ and $\tilde{y} = 1$, respectively. For future convenience, the following definition of $\tcD_0, \tcD_1$ also includes the case of an \emph{imbalanced} mixture of $\cD_0$ and $\cD_1$, parameterized by some $\eps \in (0, 1)$.

\begin{definition}[Proxy Game]
\label[definition]{def:proxy}
    Fix distributions $\cD_0, \cD_1 \in \Delta(\cX)$, a function $h : \cX \to [0, 1]$, and a parameter $\alpha \in (0, 1)$. In the \emph{proxy game}, we say that the \emph{marginal distribution} is \(\cD_\cX = (1-\alpha)\cD_0 + \alpha\cD_1\). The pair $(x, y)$ is generated by first sampling $x \sim \cD_\cX$ and then sampling $\tilde{y}|x \sim \cB(h(x))$.
    The \emph{proxy distributions} $\tcD_0, \tcD_1$ are the unique distributions such that $x|\tilde{y} \sim \tcD_{\tilde{y}}$.
\end{definition}

The following result, \cref{thm:mpv-cma-1}, shows that such proxy distributions have exactly the desired properties stated in the first part of \cref{thm:mpv-cma-succinct}, provided that $h$ is calibrated and multiaccurate.

\begin{lemma}
\label[lemma]{thm:mpv-cma-1}
    Given $\cF \subseteq \{\cX \to [0, 1]\}$, $\cD_0, \cD_1 \in \Delta(\cX)$, $\eps,\gamma \in (0, 1)$, and $h : \cX \to [0, 1]$, let $\tcD_0, \tcD_1$ be the proxy distributions, as per \cref{def:proxy} with $\alpha = 1/2$. If $h$ is an $(\cF, \eps)$-regular and $\gamma$-calibrated simulator for $g$ under $\cD_\cX$, then:
    \begin{enumerate}[(a)]
        \item $\cD_b$ and $\tcD_b$ are $(\cF, 2\eps + 2 \gamma)$-indistinguishable for each bit $b \in \{0, 1\}$,
        \item $\cD_0^{\otimes k}$ and $\cD_1^{\otimes k}$ are distinguished with advantage $\TV(\tcD_0^{\otimes k}, \tcD_1^{\otimes k}) - 8k\gamma(1+2\gamma)^k $ by
        \[
            h'(z_1, \ldots, z_k) = \bm{1}{\left[\prod_{i=1}^k h(z_i) > \prod_{i=1}^k (1-h(z_i))\right]}.
        \]
    \end{enumerate}
\end{lemma}

To see why \cref{thm:mpv-cma-1} implies the first part of \cref{thm:mpv-cma-succinct}, let $\cF$ contain all size-$s$ circuits, let $\gamma = (\eps/k)^{10}$, say, and let $h$ be the calibrated and multiaccurate simulator guaranteed by \cref{thm:ttv-calibrated}. Then, the left and right inequalities in \cref{thm:mpv-cma-succinct} follow directly from repeated applications of parts (a) and (b) of \cref{thm:mpv-cma-1}, respectively. It remains to prove \cref{thm:mpv-cma-1}.

\begin{proof}[Proof of \Cref{thm:mpv-cma-1}]
    Informally, the idea behind the proof is to first establish that
    \[
        \cD_1 = 2g\cD_\cX, \quad \cD_0 = 2(1-g)\cD_\cX, \quad \tcD_1 \approx 2h\cD_\cX, \quad \tcD_0 \approx 2(1-h)\cD_\cX,
    \]
    and then show that (a) and (b) both follow from these approximate identities. For the first two identities, observe that $\cD_1 = 2g\cD_\cX$ and $\cD_0 = 2(1 - g)\cD_\cX$ are direct consequences of the definitions of $\cD_\cX$ and $g$. To prove the other two identities, let $\hcD_0 = 2(1-h)\cD_\cX$ and $\hcD_1 = 2h\cD_\cX$. Note that $\hcD_b$, unlike $\cD_b$ and $\tcD_b$, is not necessarily a probability distribution because it need not sum to $1$. Nevertheless, we see that $\hcD_b$ differs from $\tcD_b$ by only a scaling factor, since for any $z \in \cX$,
    \[
        \tcD_b(z) = \Pr[x = z \,|\,\tilde{y} = b] = \frac{\Pr[\tilde{y} = b \,|\, x = z] \Pr[x = z]}{\Pr[\tilde{y} = b]} = \frac{\hcD_b(z)}{2 \Pr[\tilde{y}=b]}
.   \]
    It follows that $\tcD_b$ and $\hcD_b$ are close to each other in the sense that
    \begin{equation}
    \label{eq:hcd-cal}
        \sum_{z \in \cX} \bigl\lvert \tcD_b(z) - \hcD_b(z) \bigr\rvert = \sum_{z \in \cX} \bigl\lvert 2\Pr[\tilde{y}=b]-1\bigr\rvert\cdot\tcD_b(z) \le 2\gamma.
    \end{equation}
    where, in the last step, we have used the fact that $\Pr[\tilde{y}=1] = \E[h(x)] \approx_\gamma \E[g(x)] = \Pr[y=1] = 1/2$ by $\gamma$-calibration, as well as the fact that $\tcD_b$ is a distribution, which must sum to $1$. Next, using our formulas for $\cD_b$ and the definition of $\hcD_b$, we see that for all $f \in \cF$,
    \begin{equation}
    \label{eq:hcd-ma}
        \biggl\lvert\sum_{z \in \cX} f(z)\bigl(\cD_b(z)-\hcD_b(z)\bigr)\biggr\rvert = \Abs{\E_{x \sim \cD_\cX}[2f(x)(g(x) - h(x))]} \le 2\eps.
    \end{equation}
    Thus, $\cD_1$ and $\hcD_1$ are $(\cF, 2\eps)$-indistinguishable, as are $\cD_0$ and $\hcD_0$. Combining equations \eqref{eq:hcd-cal} and \eqref{eq:hcd-ma} immediately implies part (a) of the theorem.

    To prove part (b), let $h'(z)$ be the indicator for the event where $\hcD_1^{\otimes k}(z)$ exceeds $\hcD_0^{\otimes k}(z)$ for $z = (z_1, \ldots, z_k) \in \cX^k$, as in the theorem statement. Then,
    \begin{equation}
    \label{eq:mpv-hybrid}
        \biggl\lvert \sum_{z \in \cX^k} h'(z)\bigl(\hcD_1^{\otimes k}(z) - \hcD_0^{\otimes k}(z)\bigr)\biggr\rvert = \sum_{z \in \cX^k} \max\bigl(0, \hcD_1^{\otimes k}(z) - \hcD_0^{\otimes k}(z)\bigr).
    \end{equation}
    We note that any one-way restriction of $h'$ has the form $w \circ h$ for some threshold-based weight function $w : [0, 1] \to [0, 1]$ and that $\hcD_b$ is $(w \circ h, 2\gamma)$-indistinguishable from $\cD_b$ for any such $w$. Indeed, by the assumption that $h$ is $\gamma$-calibrated, the same reasoning as in equation \eqref{eq:hcd-ma} implies
    \[
        \biggl\lvert\sum_{z \in \cX} w(h(z))\bigl(\cD_b(z)-\hcD_b(z)\bigr)\biggr\rvert = \Abs{\E_{x \sim \cD_\cX}[2w(h(z))(g(x) - h(x))]} \le 2\gamma.
    \]  
    In particular, since the $\cD_b$ sums to $1$, the function $\hcD_b$ sums to at most $1 + 2\gamma$, so the function $\hcD_b^{\otimes k}$ sums to at most $(1 + 2\gamma)^k$. Therefore, by a simple $k$-step hybrid argument, we can replace each occurrence of $\hcD_b^{\otimes k}$ on the left side of equation \eqref{eq:mpv-hybrid} with $\cD_b^{\otimes k}$ while introducing an additive slack term of size at most $2k \cdot 2\gamma \cdot (1+2\gamma)^k$.

    Similarly, using equation \eqref{eq:hcd-cal}, we can replace each occurrence of $\hcD_b^{\otimes k}$ on the right side of the equation \eqref{eq:mpv-hybrid} with $\tcD_b^{\otimes k}$ while introducing an additive slack of at most $2k \cdot 2\gamma \cdot (1 + 2\gamma)^k$. Applying both this hybrid argument and the previous simplifies equation \eqref{eq:mpv-hybrid} to
    \[
        \biggl\lvert \sum_{z \in \cX^k} h'(z)\bigl(\cD_0^{\otimes k}(z) - \cD_1^{\otimes k}(z)\bigr)\biggr\rvert \ge \sum_{z \in \cX^k} \max\bigl(0, \tcD_1^{\otimes k}(z) - \tcD_0^{\otimes k}(z)\bigr) - 8k\gamma(1+2\gamma)^k,
    \]
    which is precisely the conclusion of part (b).
\end{proof}

\subsubsection{Second Part of \cref{thm:mpv-cma-succinct}}

We now seek to impose the constraint $\tcD_0 = \cD_0$ while $\tcD_1$ may still differ from $\cD_1$. As before, our construction of $\tcD_1$ will arise from the proxy game (\cref{def:proxy}). The key difference will be in the choice of the marginal distribution of $y$. In the preceding proof, we took $y \sim \cB(1/2)$. In the following proof, we will take $y \sim \cB(\alpha)$ for a small parameter $\alpha > 0$. Roughly speaking, this change ensures that the marginal distribution of $x$ is significantly tilted toward $\cD_0$, eliminating the need for an explicit proxy distribution $\tcD_0$. Because of the many similarities between the preceding result statement and the one below, we have written the key differences in blue:

\begin{lemma}
\label[lemma]{thm:mpv-cma-2}
    Given $\cF$, $\cD_0$, $\cD_1$, $\eps$, $\gamma$, and $h$ as in \cref{thm:mpv-cma-1}, let $\tcD_1$ be the proxy distribution given by \cref{def:proxy} with $\alpha={\color{navy} \eps}$. If $h$ is an $(\cF, {\color{navy} \eps^2})$-regular and $\gamma$-calibrated for $g$ under $\cD_\cX$:
    \begin{enumerate}[(i)]
        \item $\cD_1$ and $\tcD_1$ are $(\cF, {\color{navy} \eps + \gamma/\eps})$-indistinguishable,
        \item $\cD_0^{\otimes k}$ from $\cD_1^{\otimes k}$ are distinguished with advantage $\TV(\cD_0^{\otimes k}, \tcD_1^{\otimes k}) - {\color{navy} 2k(\gamma/\eps)(1+2\gamma)^k - 2k\eps}$ by
        \[
            h'(z_1, \ldots, z_k) = \bm{1}{\left[\prod_{i=1}^k h(z_i) > {\color{navy} \eps^k}\right]}.
        \]
    \end{enumerate}
\end{lemma}

In \cref{thm:mpv-cma-2}, letting $\cF$ contain all size-$s$ circuits, $\gamma = (\eps/k)^{10}$, and $h$ be the $\gamma$-calibrated and $(\cF, \eps^2)$-multiaccurate simulator guaranteed by \cref{thm:ttv-calibrated} yields the second part of \cref{thm:mpv-cma-succinct}.

\begin{proof}[Proof of \Cref{thm:mpv-cma-2}]
    Informally, the key observation is that $\cD_\cX$ is already so tilted toward $\cD_0$ that we have no need for the proxy $\tcD_0$. More formally, we first observe that $\cD_1 = g\cD_\cX / \eps$. To compute a similar formula for the proxy distribution $\tcD_1$, we set $\hcD_1 = h\cD_\cX/\eps$ and observe that for any $z \in \cX$,
    \[
        \tcD_1(z) = \Pr[x = z | \tilde{y} = 1] = \frac{\Pr[\tilde{y} = 1 \,|\, x = z] \Pr[x = z]}{\Pr[\tilde{y} = 1]} = \frac{\hcD_1(z)}{\Pr[\tilde{y}=1]/\eps}.
    \]
    It follows that $\tcD_1$ and $\hcD_1$ are close to each other in the sense that
    \begin{equation}
    \label{eq:v2-hcd-cal}
        \sum_{z \in \cX} \bigl\lvert \tcD_1(z) - \hcD_1(z) \bigr\rvert = \sum_{z \in \cX} \biggl\lvert \frac{\Pr[\tilde{y}=1]}{\eps}-1\biggr\rvert\cdot\tcD_1(z) \le \frac{\gamma}{\eps},
    \end{equation}
    where, in the last step, we have used the fact that $\Pr[\tilde{y}=1] = \E[h(x)] \approx_\gamma \E[g(x)] = \Pr[y = 1] = \eps$ by $\gamma$-calibration, as well as the fact that $\tcD_1$ is a distribution, which must sum to $1$. Next, by $(\cF, \eps^2)$-multiaccuracy, for all $f \in \cF$, 
    \begin{equation}
    \label{eq:v2-hcd-ma}
        \biggl\lvert\sum_{z \in \cX} f(z)\bigl(\cD_1(z)-\hcD_1(z)\bigr)\biggr\rvert = \Abs{\E_{x \sim \cD_\cX}\Bigl[\frac{1}{\eps}f(x)\bigl(g(x) - h(x)\bigr)\Bigr]} \le \eps.
    \end{equation}
    Thus, $\cD_1$ and $\hcD_1$ are $(\cF, \eps)$-indistinguishable. Combining equations \eqref{eq:v2-hcd-cal} and \eqref{eq:v2-hcd-ma} immediately implies part (a) of the theorem.

    To prove part (b), let $h'$ be the indicator for the event where $\hcD_1^{\otimes k}$ exceeds $\cD_\cX^{\otimes k}$, as in the theorem statement. Then,
    \begin{equation}
    \label{eq:v2-mpv-hybrid}
        \biggl\lvert \sum_{z \in \cX^k} h'(z)\bigl(\hcD_1^{\otimes k}(z) - \cD_\cX^{\otimes k}(z)\bigr)\biggr\rvert = \sum_{z \in \cX^k} \max\bigl(0, \hcD_1^{\otimes k}(z) - \cD_\cX^{\otimes k}(z)\bigr).
    \end{equation}

    By the same hybrid arguments as in the proof of \cref{thm:mpv-cma-1}, using equation \eqref{eq:v2-hcd-cal} in place of equation \eqref{eq:hcd-cal}, we can replace each occurrence of $\hcD_1^{\otimes k}$ on both the left and right sides of equation \eqref{eq:v2-mpv-hybrid} with $\cD_1^{\otimes k}$ while introducing an additive slack of at most $2k \cdot (\gamma/\eps) \cdot (1 + 2\gamma)^k$. Similarly, by our deliberately imbalanced construction of the marginal distribution, we have $\TV(\cD_\cX, \cD_0) \le \eps$, so we can similarly replace each occurrence of $\cD_\cX^{\otimes k}$ on the left and right sides of equation \eqref{eq:v2-mpv-hybrid} with $\cD_0^{\otimes k}$ while introducing an additive slack of at most $2k\cdot \eps$. Applying these two hybrid arguments simplifies equation \eqref{eq:v2-mpv-hybrid} to
    \[
        \biggl\lvert \sum_{z \in \cX^k} h'(z)\bigl(\cD_0^{\otimes k}(z) - \cD_1^{\otimes k}(z)\bigr)\biggr\rvert \ge \sum_{z \in \cX^k} \max\bigl(0, \tcD_1^{\otimes k}(z) - \tcD_0^{\otimes k}(z)\bigr) - 2k(\gamma/\eps)(1+2\gamma)^k - 2k\eps,
    \]
    which is precisely the conclusion of part (b). 
\end{proof}

\subsection{Characterization via Supersimulators}
\label[section]{sec:mpv-super}

In this section, we prove \cref{thm:mpv-super-succinct}, which fully closes the complexity gap between the upper and lower bounds in \cref{thm:mpv-cma-succinct}. We recall the theorem statement below. Its proof is extremely short, given the work we have already done to establish \cref{thm:mpv-cma-1,thm:mpv-cma-2}. The only remaining step is to plug in the supersimulator provided by \cref{thm:super}.

\mpvsuper*

\begin{proof}
    In the statement of \cref{thm:mpv-cma-1}, whenever $h$ is computable by a Boolean circuit of size $s_h$ and produces $b$-bit outputs, the corresponding $h'$ has complexity at most $s_{h'} = ks_h + (kb)^{O(1)}$. Thus, given an initial size $s \in \N$, we seek a $\gamma$-calibrated, $\Omega(\gamma)$-precision simulator $h$ of complexity $s_h$ that fools distinguishers of size $\max(s, s_{h'})$ with error $\eps$, where $\gamma = (\eps/k)^{10}$. Let $\cG(h)$ comprise all such distinguishers. Then, by \cref{thm:super},\footnote{See the discussion surrounding \cref{thm:supersimulator-basic} for more detail on incorporating $\gamma$-calibration with $\Omega(\gamma)$-precision.} there exists a suitable supersimulator $h$ with respect to these parameters of size at most $k^{O(1/\eps^2)}s$. Finally, observe that starting from \cref{thm:mpv-cma-2} instead of \cref{thm:mpv-cma-1} doubles the leading exponent on $\eps$ from $2$ to $4$ due to its $\eps^2$-regularity requirement.
\end{proof}

\bibliographystyle{alpha}
\bibliography{sources}

\appendix

\section{Constructon of Supersimulators}
\label[appendix]{sec:supersimulators}

We begin with a simple, unified proof of \cref{thm:ttv,thm:super} (complexity-theoretic regularity and supersimulators). Since these results are stated in terms of a single projection operation, rather than a nested sequence of such operations, we must slightly modify the proof from \cite{dwork2021outcome}. We do so using the following lemma about \emph{prefix sums}.

\begin{lemma}[Prefix Sums]
\label[lemma]{thm:ftrl}
    If $a_0, \ldots, a_k \in \R$ and $s_j = [a_0 + \cdots + a_j]_0^1$ and $b \in [0, 1]$, then
    \[
        \sum_{j=1}^k a_j(b - s_j) \le \frac{1}{2}(b-a_0)^2.
    \]
\end{lemma}

\begin{proof}
    We induct on $k$. The base case ($k = 0$) is trivial. For $k \ge 1$, we observe that
    \[
        \sum_{j=1}^k a_j(b - s_j) = \sum_{j=1}^{k-1} a_j(s_k - s_j) + \sum_{j=1}^k a_j(b - s_k).
    \]
    By the inductive hypothesis, the first sum on the right side is at most $(s_k-a_0)^2/2$. Thus, we must show that rightmost sum is at most $((b-a_0)^2 - (s_k-a_0)^2)/2$. Some algebra does the trick:
    \[
        \sum_{j=1}^k a_j(b - s_k) = \frac{1}{2}\Bigl((b-a_0)^2 - (s_k-a_0)^2\Bigr) + \frac{1}{2}\biggl(\underbrace{\Bigl(s_k - \sum_{j=0}^k a_j\Bigr)^2 - \Bigl(b - \sum_{j=0}^k a_j\Bigr)^2}_{(*)}\biggr),
    \]
    where the term $(*)$ is $\le 0$ by the definition of $s_k$ and the assumption $b \in [0, 1]$.
\end{proof}

Equipped with this result on prefix sums, we are now ready to prove \cref{thm:super}.

\supersimulators*

\begin{proof}
    Given $h_0$, we will inductively define functions $h_1, \ldots h_{k-1}: \cX \to [0, 1]$ and then argue that some $h_j$ must be a $(\cG(h_j), \delta)$-regular simulator for $g$. For $j \ge 1$, suppose that $h_{j-1}$ is \emph{not} $(\cG(h_{j-1}), \delta)$-regular, and let $f_j$ be any function in $\pm \cG(h_{j-1})$ such that
    \[
        \E\Bigl[f_j(x)\bigl(g(x) - h_{j-1}(x)\bigr)\Bigr] > \delta,
    \]
    where the expectation is over $x \sim \cD$. Next, fix $\eta > 0$ and define
    \[
        h_j(x) = \Bigl[h_0(x) + \eta \sum_{i=1}^j f_i(x)\Bigr]_0^1.
    \]
    Since $\abs{f_j(x)} \le 1$, the functions $h_j$ and $h_{j-1}$ differ pointwise by at most $\eta$. Thus,
    \[
        \E\Bigl[f_j(x)\bigl(g(x) - h_j(x)\bigr)\Bigr] > \delta - \eta.
    \]
    Applying \cref{thm:ftrl} with $b = g(x)$ and $a_0 = h_0(x)$ and $a_j = \eta f_j(x)$ for each index $j \in [k]$ yields
    \[
        k \cdot \eta(\delta - \eta) < \E\Bigl[\sum_{j=1}^k \eta f_j(x)\bigl(g(x) - h_j(x)\bigr)\Bigr] \le \frac{1}{2}\E\Bigl[\bigl(g(x) - h_0(x)\bigr)^2\Bigr] \le \frac{1}{2}.
    \]
    Setting $\eta = \delta / 2$ and $k = 2/\delta^2$ yields a contradiction. We conclude that one of the functions $h_j$ for $0 \le j < k$ must have been an $(\cG(h_j), \delta)$-regular simulator for $g$, as desired.
\end{proof}

Besides the simpler projection operation, which was needed for one of our downstream applications, the preceding proof is largely the same as the proof of Lemma 5.6 of \cite{dwork2021outcome} regarding code-access outcome indistinguishability. That proof, in turn, shares much in common with the boosting-style proof of \cref{thm:ttv} from \cite{trevisan2009regularity}, with the key difference being that one now allows the family of candidate distinguishers to expand dramatically at each step. Interestingly, a modification like this was made in the graph regularity context to give an alternate proof of Szemer\'{e}di's regularity lemma---see Theorem 3.2 of \cite{lovasz2007analyst}, in which the distinguisher family consists of arbitrary unions of increasingly many rectangles.

We also remark that our ``prefix sum lemma'' in \cref{thm:ftrl} is essentially a special case of the analysis of standard algorithms in online learning and convex optimization, such as Follow-the-Regularized-Leader (FTRL) and mirror descent with lazy projections. Rather than invoking the full strength of this machinery, we have chosen to isolate the minimal technical tool required for the job in \cref{thm:ftrl}. For more on online convex optimization, we refer the reader to \cite{bubeck2015convex,mcmahan2017online}.

\subsection{Interpretation of the Construction}

\Cref{thm:super} shows how to construct a simulator $h$ that fools distinguishers whose complexity exceeds that of $h$ by an arbitrary, prespecified \emph{growth} function. To understand what this means, recall from \cref{sec:preliminaries} that for a given function family $\cF$, we measure the complexity of $h \in \cF_{(s_1, s_2)}$ by two numbers $s_1$ and $s_2$, which count the number of calls to functions in $\cF$ and the number of additional circuit gates, respectively, required to compute $h$.

Accordingly, a natural special case of the generic notion of a ``growth function'' $\cG$ in \cref{thm:super} is that of a nondecreasing function $G : \N^2 \to \N^2$ under the partial ordering of $\N^2$ in which $s' \ge s$ if both $s'_1 \ge s_1$ and $s'_2 \ge s_2$. In this slightly more concrete notation, \cref{thm:super} has the following corollary.

\begin{corollary}
\label[corollary]{thm:supersimulator-basic}
    For all distributions $\cD \in \Delta(\cX)$, distinguisher families $\cF \subseteq \{\cX \to [0, 1]\}$, target functions $g : \cX \to [0, 1]$, error tolerances $\eps \in (0, 1/2)$, and nondecreasing $G : \N^2 \to \N^2$, there exists a size bound $s \in \N^2$ and a simulator $h \in \cF_s$ such that:
    \begin{itemize}
        \item \textbf{(regularity)} $h$ is $\bigl(\cF_{G(s)}, \eps\bigr)$-regular,
        \item \textbf{(complexity)} $s \le S_{O(1/\eps^2)}$, where $S_0 = (1, 1)$ and $S_{i+1} = S_i + G(S_i) + \bigl(0, (\log(1/\eps))^{O(1)}\bigr)$.
    \end{itemize}
\end{corollary}

In \cref{thm:supersimulator-basic}, although $h \in \cF_s$, it fools any distinguisher in the class $\cF_{G(s)}$, which can be much larger than $\cF_s$ for appropriately chosen growth functions $G(s) \gg s$. For example, in the Boolean circuit setting, letting $\cF$ consist of the $n$ coordinate functions $x \mapsto x_i$ for $x \in \{0, 1\}^n$ and defining $G(s_1, s_2) = (n, \max(n, s_2)^k)$ yields \cref{thm:super-succinct}, the version of the lemma that we discussed in the introduction. Note also that the complexity of $h$ relative to $\cF$ is still bounded above by a quantity independent of the target function $g$. Indeed, both $s_1$ and $s_2$ can be bounded above by a constant that depends only on the growth function $G$ and the error tolerance $\eps$.

There are several important consequences of being capable of fooling distinguishers sufficiently larger than oneself. For example, one can check that if a simulator $h \in \cF_{s}$ is $(\cF_{s'}, \eps)$-regular, where $s_1' \ge s_1$ and $s_2' \ge s_2 + 1/\eps^{O(1)}$, then $h$, rounded to integer multiples of $\eps$, is automatically $O(\eps)$-calibrated. Roughly speaking, this is because a distinguisher in $\cF_{s'}$ has enough circuit gates to first compute the value $h(x)$ itself and then use its additional $1/\eps^{O(1)}$ gates to evaluate a weighted calibration test (see \cref{sec:preliminaries}). Similarly, if $s_1' \ge s_1 + 1$ and $s_2' \ge s_2 + 1/\eps^{O(1)}$, then the $\eps$-rounded version of $h$ must be $(\cF, \eps^{O(1)})$-multicalibrated. An analogous statement can be made for oracle-access outcome indistinguishability \cite{dwork2021outcome}, in which distinguishers are allowed to make $q$ oracle calls to $h$, if $s_1'$ and $s_2'$ are at least $q \cdot s_1$ and $q \cdot s_2$ respectively.

\subsection{Alternate Construction via Iteration}

In this section, we present a new and slightly more general supersimulator construction, which involves not only an arbitrary growth function $G$ but also a decaying error function $\eps : \N^2 \to (0, 1)$. Our construction is inspired by the iteration technique from the graph regularity literature discussed in \cref{sec:preliminaries}, in which one takes a sequence of increasingly regular vertex partitions. In the graph context, ``increasingly regular'' means shrinking the error parameter $\eps \to 0$ but usually means keeping the distinguisher class (cuts) fixed. In the abstract setting, ``increasingly regular'' could be interpreted as either shrinking $\eps$ or expanding $\cF$, or both. Our first construction (\cref{thm:super,thm:supersimulator-basic} above) can be viewed as expanding $\cF$ while keeping $\eps$ fixed. Our second construction (\cref{thm:supersimulator-shrinking} below) expands $\cF$ and shrinks $\eps$ simultaneously. It can be viewed as a straightforward iteration of \cref{thm:ttv-calibrated}, the calibrated version of complexity theoretic-regularity. Although the following version is more arguably more general than the previous, the previous had a simpler proof, which is why we presented it separately.

\begin{theorem}[Supersimulators, Expanding/Shrinking]
\label[theorem]{thm:supersimulator-shrinking}
    For all $\cD \in \Delta(\cX)$, $\cF \subseteq \{\cX \to [0, 1]\}$, $g : \cX \to [0, 1]$, nonincreasing $\eps : \N^2 \to (0, 1/2)$, $\alpha \in (0, 1/2)$, and nondecreasing $G : \N^2 \to \N^2$, there exist $s, s' \in \N^2$ and $h \in \cF_s$ and $h' \in \cF_{s'}$ such that:
    \begin{itemize}
        \item \textbf{(similarlity)} $\E_{x \sim \cD}(h(x) - h'(x))^2 \le \alpha + O(\eps(s))$,
        \item \textbf{(regularity)} $h'$ is $(\cF_{G(s)}, \eps(s))$-regular,
        \item \textbf{(complexity)} $s, s' \le S_{\lfloor 1/\alpha \rfloor}$, where $S_0 = (1, 1)$ and \[S_{i + 1} \le O\bigl(1/\eps(S_i)^{2}\bigr)G(S_i)  + \Bigl(0, \, \tilde{O}\bigl(1/\eps(S_i)^3\bigr)\Bigr).\]
    \end{itemize}
\end{theorem}

\Cref{thm:supersimulator-shrinking} provides \emph{two} functions, $h$ and $h'$. It guarantees that that the latter is a simulator that fools distinguishers much more complex than the former. Specifically, while $h \in \cF_s$, the simulator $h'$ fools all distinguishers in $\cF_{G(s)}$, where $G$ is the arbitrary growth function. Moreover, it fools them extremely well, allowing only a vanishingly small distinguishing error of $\eps(s)$, which we may take to decay arbitrarily fast with $s$.

The fact that $h'$ fools distinguishers more complex than $h$ is only useful if we know that $h$ is nontrivial. This is achieved by the \emph{similarity} condition, which states that $h$ and $h'$ are similar to each other in $L^2$ norm. One consequence of this condition is that $h$ itself is a simulator that fools distinguishers larger than itself, albeit with an error parameter that is not vanishingly small:

\begin{corollary}
    Let $\cF, G, \alpha, \eps, h, s$ be as in \cref{thm:supersimulator-shrinking}. Then $h$ is $(\cF_{G(s)}, O(\sqrt{\alpha + \eps(s)}))$-regular.
\end{corollary}

\begin{proof}
    Fix any $f \in \pm \cF_{G(s)}$. Then, by Cauchy-Schwarz,
    \[
        \E\bigl[ f(x)(g(x) - h(x)) \bigr] \le \E\bigl[ f(x)(g(x) - h'(x)) \bigr] + \sqrt{\E\bigl[f(x)^2\bigr] \cdot \E\bigl(h(x)-h'(x)\bigr)^2}.
    \]
    To bound the first term, apply the $(\cF_{G(s)}, \eps(s))$-regularity of $h'$. For the second term, note that $f$ takes values in the range $[-1, +1]$ and apply the similarity condition in \cref{thm:supersimulator-shrinking}.
\end{proof}

As already mentioned, the proof of \cref{thm:supersimulator-shrinking} involves a simple iteration of the calibrated version of the complexity-theoretic regularity lemma (\cref{thm:ttv-calibrated}). We use this calibrated version, rather than the default version (\cref{thm:ttv}), in order to establish the $L^2$ similarity condition. Without the calibration condition, we would only be able to upper bound the signed potential difference between $h$ and $h'$, rather than their $L^2$ distance.

\begin{proof}[Proof of \Cref{thm:supersimulator-shrinking}]
    Choose any $h_0 : \cX \to [0, 1]$. For each $i \in \N$, let $h_{i+1}$ be the $(\cF_{G(S_i)}, \eps(S_i))$-regular and $\eps(S_i)$-calibrated predictor that \cref{thm:ttv-calibrated} guarantees lies in
    \[
        (\cF_{G(S_i)})_{\Bigl(O\bigl(1/\eps(S_i)^2\bigr), \, \tilde O\bigl(1/\eps(S_i)^3\bigr)\Bigr)}.
    \] Expanding each call to a function in $\cF_{G(S_i)}$ with calls to functions in $\cF$, we see that $h_{i +1} \in \cF_{S_{i+1}}$ for the specified size bound $S_{i+1}$. Next, define the potential function $\Phi(i) = \E(g(x) - h_i(x)^2)$ as in the proof of \cref{thm:supersimulator-basic}. Similar algebra shows that the $L^2$ distance between $h_i$ and $h_{i+1}$ relates to the difference between $\Phi(i)$ and $\Phi(i+1)$ via:
    \[
         \E(h_i(x) - h_{i+1}(x))^2 = \Phi(i) - \Phi(i+1)- 2\E(h_i(x) - h_{i+1}(x))(g(x) - h_{i+1}(x)).
    \]
    Recall that $(\cF_{G(s)}, \eps(s))$-regularity of $h_{i+1}$ implies \[\Bigl\lvert\E\bigl[h_i(x)(g(x) - h_{i+1}(x))\bigr]\Bigr\rvert \le \eps(s).\] Similarly, $\eps(s)$-calibration of $h_{i+1}$ implies \[\Bigl\lvert\E[h_{i+1}(x)(g(x) - h_{i+1}(x))]\Bigr\rvert \le \eps(s).\] To conclude, choose an $i \le 1/\alpha$ satisfying $\Phi(i) - \Phi(i + 1) \le \alpha$, and let $(h, h') = (h_{i}, h_{i + 1})$.
\end{proof}

\begin{figure}[b]
     \centering
     \begin{subfigure}[b]{0.3\textwidth}
         \centering
         \includegraphics[width=\textwidth]{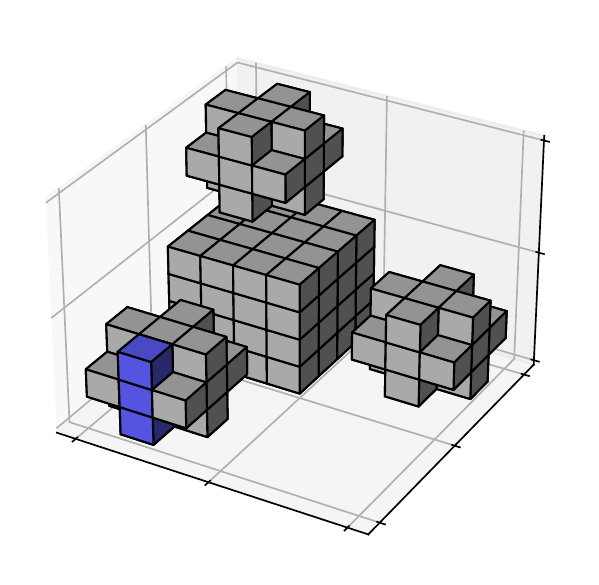}
         \caption{Tester $T$}
         \label{fig:tester}
     \end{subfigure}
     \begin{subfigure}[b]{0.3\textwidth}
         \centering
         \includegraphics[width=\textwidth]{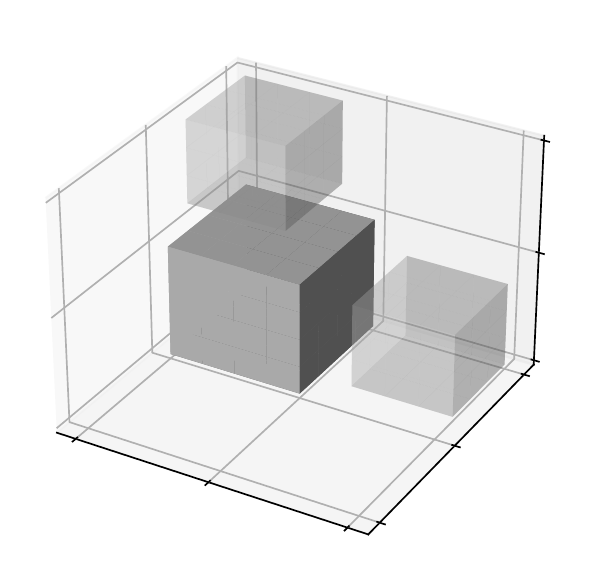}
         \caption{Simulator $\widetilde{T}_j$}
         \label{fig:before}
     \end{subfigure}
     \begin{subfigure}[b]{0.3\textwidth}
         \centering
         \includegraphics[width=\textwidth]{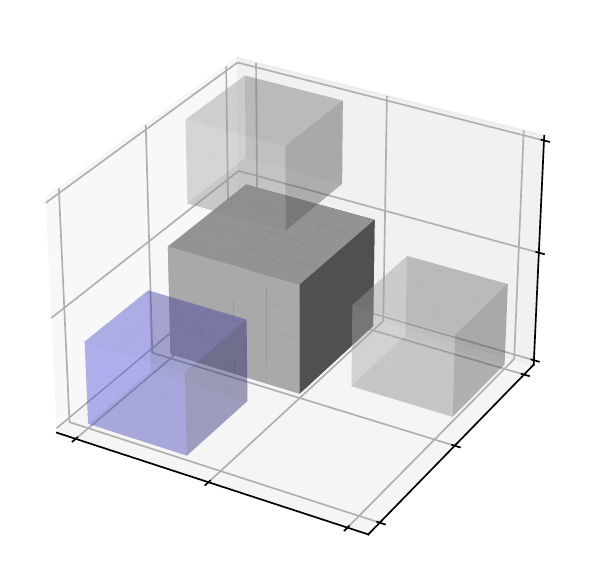}
         \caption{Simulator $\widetilde{T}_{j+1}$}
         \label{fig:after}
     \end{subfigure}
    \caption{Rough illustration of the supersimulator construction used in \cref{sec:symmetry-thm}. Plot (\subref{fig:tester}) depicts a deterministic $3$-sample tester $T$ by the set of triples that cause it to output \textsc{Accept}. A one-way restriction of $T$ is shown in blue. Plot (\subref{fig:before}) depicts a simulator in the sequence, with translucent regions indicating fractional estimates for $T$'s \textsc{Accept} region. Plot (\subref{fig:after}) depicts the next simulator in the sequence, after an update has been performed based on the chosen restriction of $T$.}
    \label{fig:super-visual}
\end{figure}

\end{document}